\documentclass[12pt]{elsarticle}

\usepackage{amsmath}
\usepackage{amsfonts}%
\usepackage{amssymb}%
\usepackage{graphicx}
\usepackage[margin=1in]{geometry}
\usepackage{float}
\usepackage{verbatim}
\usepackage{amsthm}
\usepackage{amssymb}
\usepackage{esint}
\usepackage{slashbox}
\usepackage{color,soul}
\usepackage[font={small,it}]{caption}

\numberwithin{equation}{section}

\providecommand{\U}[1]{\protect\rule{.1in}{.1in}}
\newtheorem{theorem}{Theorem}[section]

\newtheorem{definition}[theorem]{Definition}

\newtheorem{lemma}[theorem]{Lemma}

\newtheorem{proposition}[theorem]{Proposition}
\newtheorem{remark}[theorem]{Remark}

\begin{document}

\title{Robust Recovery of Stream of Pulses using Convex Optimization}
\author[add1]{Tamir Bendory}
\ead{stdory@campus.technion.ac.il}
\author[add2]{Shai Dekel}
\ead{shai.dekel@ge.com}
\author[add1]{Arie Feuer}
\ead{feuer@ee.technion.ac.il}

\address[add1]{Department of Electrical Engineering, Technion -- Israel Institute of Technology}
\address[add2]{School of mathematical sciences, Tel--Aviv University and GE Global Research}

\date{}
\begin{abstract}

This paper considers the problem of recovering the delays and amplitudes 
of a weighted superposition of pulses.  This problem is motivated by a variety of applications, such as ultrasound and radar.
We show that for univariate and bivariate stream of pulses, one can
recover the delays and weights to any desired accuracy by solving
a tractable convex optimization problem, provided that a pulse-dependent
separation condition is satisfied. The main result of this paper states that the recovery is robust to additive noise
or model mismatch. 
\end{abstract}
\maketitle

\section{Introduction\label{Section1}}

In this paper we consider signals of the form
\begin{equation}
y(t)=\sum_{m}c_{m}K_{\sigma}\left(  t-t_{m}\right)  ,\label{1}%
\end{equation}
where $K_{\sigma}(t):=K(\sigma^{-1}t)$ and $\sum_m\vert c_m\vert <\infty$. We assume that the kernel (pulse) $K$ and the scaling $\sigma>0$ are known, whereas the
delays $\left\{  t_{m}\right\}  $ and the real amplitudes $\left\{
c_{m}\right\}  $ are unknown.  The delayed versions of the kernel, $\left\{
K_{\sigma}(t-t_{m})\right\}  _{m}$, are often referred to as \emph{atoms}. In Section \ref{Section2} we discuss specific
requirements for $K$.

An alternative representation of the signal in (\ref{1}) is%
\[
y(t)=\left(  K_{\sigma}\ast x\right)  (t)=\int_{\mathbb{R}}K_{\sigma
}(t-s)dx(s),
\]
where 
\begin{equation}
x(t)=\sum_{m}c_{m}\delta_{t_{m}}(t), \label{2}%
\end{equation}
and $\delta_t$ denotes a Dirac measure. 
This model is quite common in a number of engineering applications such as
ultrasound \cite{tur2011innovation,wagner2012compressed,bendory2015stable}, radar
\cite{Bar-IlanSub-Nyquis} and more (see e.g. \cite{vetterli2002sampling,dragotti2007sampling,schiebinger2015superresolution}). In these applications, we transmit a pulse and measure the received echoes. 
This formulation resembles the problem of super-resolution that has received considerable attention recently \cite{candes2013towards,candes2013super,fernandez2015super,mishra2015}. However, our problem is defined on $\mathbb{R}^n$ whereas the super-resolution problem is defined on the n-dimensional torus. Additionally, we are not restricted to band-limited kernels and can deal with a broader family of convolution kernels as presented in Definition \ref{Definition2}.  

Let $\hat{x}$, $\hat{y}$ and $\hat{K}_{\sigma}$ be the Fourier transforms of $x,y$ and $K_{\sigma}$, respectively. A naive approach will aim to recover $x$ through the relation $\hat{x}=\hat{y}/\hat{K}_{\sigma}$. However, since $x$ is a Dirac sequence and consequently $\hat{x}$ is non-vanishing, this approach will not work for practical decaying kernels even if $\hat{K}_{\sigma}$ is non-vanishing. 
A well-known approach to decompose the signal into its atoms is by using parametric methods such as MUSIC and matrix pencil \cite{stoica2005spectral,hua1990matrix,roy1989esprit,schmidt1986multiple}. These methods do not assume any structure on the signal besides sparsity.
However, they tend to
be unstable in the presence of noise or model mismatch due to sensitivity of polynomial root finding.
As far as we know, they have no guarantees for their robustness (although significant progress has been made recently for the particular case of super-resolution \cite{moitra2014threshold,liao2014music}). 

An alternative way is to utilize compressed sensing and sparse representations theorems, relying on the sparsity of the signal in a discrete basis (e.g. \cite{elad2010sparse,donoho2006compressed}).
 Evidently, signals that have sparse representation in a
continuous dictionary might not have sparse representation after
discretization \cite{chi2011sensitivity}. An obvious
technique to alleviate this basis mismatch is by fine discretization.
However, the aforementioned fields cannot explain the success of $\ell_{1}$
minimization or greedy algorithms as the
dictionaries have high coherence \cite{just_dis}.

Recently, a number of works suggested sparsity-promoting convex optimization
techniques over the continuum. In these works, the notion of coherence does not play any
role. Particularly, it was suggested to recast Total-Variation (TV) and atomic
norm minimization as semi-definite programs (SDP) in order to recover point sources
from low-resolution data on the line
\cite{candes2013towards,candes2013super}, and on the sphere \cite{bendory2013exact,bendory2015super,bendory2015recovery}, or for line spectral estimation \cite{bhaskar2011atomic,tang2012compressive,tang2013near}. Similar approach was applied to
the recovery of non-uniform splines from their projection onto algebraic
polynomial spaces \cite{bendory2013Legendre,de2014non} (see also
\cite{de2012exact,azais2015spike}).

We follow this line of work and employ TV minimization techniques to recover
the delays and amplitudes from the measured data $y(t)$. The aforementioned works exploit the special structure of trigonometric and algebraic polynomials to recast the dual of the infinite dimensional problem as a finite SDP, based on the Riesz-Fejer Thereom \cite{dumitrescu2007positive} (see also an SDP approach for the primal problem in \cite{de2015exact}).  
 However, as far as we know, there is no tractable algorithm solving the TV minimization on $\mathbb{R}$. Therefore, in practice, we consider a sampled version of $y(t)$, which reduces the infinite
dimensional TV minimization to a finite $\ell_{1}$ minimization problem. The
continuous nature of our analysis guarantees that the results hold for any
discretization of the parameter space, with no dependence on the dictionary coherence (see also \cite{tang2013atomic}).  
 The behavior of the discrete optimization problem solution, when the underlying signal is defined on the
continuum (as typically occurs in practical applications) is analyzed in detail in \cite{duval2015exact,duval2015sparse,just_dis}.

The outline of this paper is as follows. After this introduction section, Section \ref{Section2} presents our main results. In Theorems \ref{Theorem1} and \ref{th:2d}
we show that a sufficient condition for a successful recovery  is
that the delays are sufficiently separated. We further establish in Theorem
\ref{Theorem2} that recovery via $\ell_{1}$ minimization is robust to
additive noise or model mismatch, and the error is proportional to the noise level. We note that we make no statistical assumptions on the noise so our results are quite general. 
Section \ref{Section3}, which is the main body of this paper, is dedicated to the presentation of
proofs. In Section \ref{Section4} we present the
results of numerical experiments we have conducted, and in Section \ref{Section5} we draw conclusions and discuss the implications of our results.

\section{Main Results\label{Section2}}

In this section we present our main results. For the benefit of the reader we
provide the formal definition of Total Variation (TV) norm (see
\cite{rudin1986real}):

\begin{definition}
 Let $\mathcal{B}(A)$ be the Borel $\sigma$-Algebra on a space $A$, and denote by $\mathcal{M}(A)$ the associated space of real
Borel measures. The {Total Variation} \emph{(TV)} of a real Borel measure
$v\in\mathcal{M}(A)$ over a set $B\in\mathcal{B}(A)$ is defined by
\[
|v|(B)=\sup\sum_{k}|v(B_{k})|,
\]
where the supremum is taken over all partitions of $B$ into countable disjoint subsets. The total variation $|v|$ is a non-negative
measure on $\mathcal{B}(A)$, and the {Total Variation (TV) norm} of $v$
is defined as
\[
\Vert v\Vert_{TV}=|v|(A).
\]

\end{definition}

The TV norm of a signed measure can be interpreted as
the generalization of $\ell_{1}$ norm to the real line. For a discrete measure of the
form of (\ref{2}), it is easy to see that
\[
\Vert x\Vert_{TV}=\sum_{m}|c_{m}|.
\]

\noindent Next we need the following definitions:

\begin{definition}
\label{Definition1} A set of points $T\subset\mathbb{R}$ is said to satisfy
the minimal separation condition for a kernel dependent $\nu>0$ and a given
$\sigma>0$ if%
\begin{equation*}
\Delta:=\min_{t_{i},t_{j}\in T,t_{i}\neq t_{j}}\left\vert t_{i}-t_{j}%
\right\vert \geq\nu\sigma.
\end{equation*}

\end{definition}

\begin{definition}
\label{Definition2} A kernel $K$ is {admissible} if it
has the following properties:

\begin{enumerate}
\item $K\in \mathcal{C}^3\left(\mathbb{R}\right)  $, is real and even.

\item \underline{Global property:} There exist constants $C_{\ell}>0, \ell=0,1,2,3$ such
that $\left\vert K^{\left(  \ell\right)  }\left(  t\right)  \right\vert
\leq\ C_{\ell} / \left( {1+t^{2}} \right)$ , where $ K^{\left(  \ell\right)  }\left(  t\right)$ denotes the $\ell^{th}$ 
derivative of  $K$.

\item \underline{Local property:} There exist constants $\varepsilon,\beta>0$ such that 
\begin{enumerate}
\item  $K(t)>0$ for all $\vert t\vert \leq \varepsilon$ and $K(t)  < K(\varepsilon)$ for all $\vert t\vert>\varepsilon$,
\item  $K^{\left(2\right)  }\left( t\right)  <-\beta$  for all $\vert t\vert \leq \varepsilon$.
\end{enumerate}

\end{enumerate}
\end{definition}

\begin{remark}
\label{Remark3}  By the Taylor Remainder theorem, for any $0<t \le \varepsilon$, there exists $0<\xi \le t$ such that 
\begin{equation*}
K^{(1)}(t)=K^{(1)}(0)+tK^{(2)}(\xi)=tK^{(2)}(\xi)<0.
\end{equation*} 
Hence, the local property implies that $K$ is monotonically decreasing in $t\in(0,\varepsilon)$.
\end{remark}

\begin{remark}
\label{Remark1} The global property in Definition \ref{Definition2} can be
somewhat weakened to $\left\vert K^{\left(  \ell\right)  }\left(  t\right)
\right\vert \leq C_{\ell} / (1+\vert t\vert ^{1+s})$ for some $s > 0$. In this case, the separation condition would become dependent on
$s$.
\end{remark}

\begin{remark}
\label{Remark4} Two prime examples for an admissible kernel are the Gaussian kernel, $K\left(  t\right)
=e^{-t^{2}/2}$ and the Cauchy kernel, $K(t)= 1/\left({1+t^2}\right)$.
The reader can readily verify that both are admissible. Table \ref{table1} presents the numerical constants $C_\ell$, $\ell=0,1,2,3$ for both kernels. 
\end{remark}

Our first theorem states that one can recover the unknown delays and amplitudes of a given stream of pulses
(\ref{1}) by simply minimizing the TV norm of (\ref{2}).

\begin{theorem}
\label{Theorem1}
For any admissible kernel $K$, there exists $\nu>0$ such that for any $\sigma>0$, the delays and amplitudes of any signal $y$
of the form (\ref{1}), with delays $T:=\{t_{m}\}$ satisfying the
separation condition of Definition \ref{Definition1}, can be recovered as the unique solution of
\begin{equation}
\min_{\tilde{x}\in\mathcal{M}(\mathbb{R})}\left\Vert \tilde{x}\right\Vert
_{TV}\quad \mbox{subject to} \quad y(t)=\int_{\mathbb{R}}K_{\sigma}(t-s)d\tilde{x}\left(
s\right)  ,\label{4}%
\end{equation}
where $\mathcal{M}(\mathbb{R})$ is the space of signed Borel measures on
$\mathbb{R}$.
\end{theorem}

\begin{remark}
We note that the feasible set of  the optimization problem (\ref{4}) is one-element if and only if the Fourier transform of $K_\sigma$ is not vanishing identically on an interval. In this sense, Theorem \ref{Theorem1} is trivial for  kernels whose Fourier transform does not vanish identically on an interval, such as the Gaussian and Cauchy kernels and the signal can be recovered linearly by  $\hat{y}=\hat{K}_\sigma\hat{x}$, where  $\hat{x}$, $\hat{y}$ and $\hat{K}_{\sigma}$ are the Fourier transforms of $x,y$ and $K_{\sigma}$, respectively. In this manuscript, we use Theorem  \ref{Theorem1} as a step towards the main result, which is the robustness under noise or model mismatch as presented in Theorem \ref{Theorem2}. We emphasize that from practical point of view, the stability is of crucial importance as discussed later in this section and demonstrated in Figure \ref{fig:LS}.
\end{remark}

\begin{table}
\begin{center}

    \begin{tabular}{| l | l | l |}
    \hline
   \backslashbox{}{Kernels} & Gaussian $:=e^{\frac{-t^2}{2}}$  & Cauchy $:=\frac{1}{1+t^2}$   \\ \hline
    $C_0$ & 1.22 & 1 
 \\ \hline
  $C_1$ & 1.59 & 1 \\ \hline
   $C_2$ & 2.04 & 2 \\\hline
      $C_3$ & 2.6 & 5.22 \\ \hline
      $K^{(2)}(0)$ & -1 & -2 \\ \hline      
            empirical $\nu$ & 1.1 & 0.45 \\ \hline      
    \end{tabular}
    \end{center}

    \caption{The table presents the numerical constants of the global property in Definition \ref{Definition2} for our two prime examples, Gaussian and Cauchy kernels. Additionally, we evaluated by numerical experiments the minimal empirical value of $\nu$, the separation constant of Definition \ref{Definition1} for each kernel (see also Figure \ref{fig1}).  }
\label{table1}
\end{table}

Table \ref{table1} shows the numerical constants associated with the Gaussian and Cauchy kernels. Moreover, it presents the minimal empirical values of the separation constant $\nu$ (see Definition \ref{Definition1}), as evaluated numerically (see Figure \ref{fig1}). The proof in Section \ref{sec:proof1} reveals the dependence of $\nu$ on the nature of the admissible kernel. For instance, (\ref{estimate-nu-1}) shows that small $\vert K^{(2)}(0)\vert$ (that is, flatness near the origin) requires a larger separation constant $\nu$ (see also e.g. equations (\ref{estimate-nu-2}),(\ref{estimate-nu-3})).

We want to emphasize that our model can be extended to other types of underlying signals, not necessarily a spike train as in (\ref{2}). For instance, suppose that the underlying signal itself is a stream of pulses of the form $\sum_mc_m\grave{K}_\sigma(t-t_m)$. In this case, the measurements are given as
$y(t)=\left(K_\sigma\ast\grave{K}_\sigma\ast x\right)(t)$, 
where $x(t)$ is a signal of the form of (\ref{2}). Therefore, our results hold immediately if the convolution kernel
$\tilde{K}(t)=\left({K}\ast \grave{K}\right)(t)$ meets the definition of admissible kernel and $T$ satisfies the associated separation condition.
For instance, if K and $\grave{K}$ are both Gaussian kernels with standard deviations of $\sigma_1$ and $\sigma_2$, then $\tilde{K}$ is also
Gaussian with standard deviation of $\sigma=\sqrt{\sigma_1^2+\sigma_2^2}$ and thus obeys the definition of admissible kernel. 
 
The univariate result can be extended to bivariate signals. Consider a signal of the form 
\begin{equation} \label{eq:2d_signal}
y\left(t,u\right)=\sum_{m}c_{m}K_{2,\sigma}\left(t-t_{m},u-u_{m}\right),
\end{equation}
where $K_{2,\sigma}:=K_{2}\left(\sigma_{t}^{-1}t,\sigma_{u}^{-1}u\right),\thinspace \sigma_{t},\sigma_{u}>0$ and $K_2$ is a bivariate kernel.
The following are the equivalent of Definitions \ref{Definition1} and \ref{Definition2} in the bivariate case:

\begin{definition}
\label{def:separation_2d} A set of points $T_{2}\subset\mathbb{R}^{2}$
is said to satisfy the minimal separation condition for a kernel dependent
$\nu>0$ and a given $\sigma_{t},\sigma_{u}>0$ if 
\[
\Delta:=\min_{\left(t_{i},u_{i}\right),\left(t_{j},u_{j}\right)\in T_{2},i\neq j}\max\left\{ \frac{\vert t_{i}-t_{j}\vert }{\sigma_{t}},\frac{\vert u_{i}-u_{j}\vert }{\sigma_{u}}\right\} \geq\nu.
\]
\end{definition}

\begin{definition}
\label{def:2d_kernel}A bivariate kernel $K_{2}$ is {admissible} if it has the following properties:
\begin{enumerate}
\item $K_{2}\in \mathcal{C}^3\left(\mathbb{R}^2\right)$, is real and even, that is 
\[
K_{2}\left(t,u\right)=K_{2}\left(-t,u\right)=K_{2}\left(t,-u\right)=K_{2}\left(-t,-u\right).
\]

\item \underline{Global property}: There exist constants $C_{\ell_{1},\ell_{2}}>0$
such that $\left|{K_{2}^{(\ell_1,\ell_2)}(t,u)}\right|\le\frac{C_{\ell_{1},\ell_{2}}}{\left(1+t^{2}+u^{2}\right)^{3/2}}$, for $\ell_1+\ell_2\leq 3$, where $K_{2}^{(\ell_1,\ell_2)}(t,u):=\frac{\partial^{\ell_1}\partial^{\ell_2}}{\partial t^{\ell_1}\partial u^{\ell_2}}K_{2}\left(t,u\right)$.
\item \underline{Local property}: There exist constants $\beta,\varepsilon>0$ such that
\begin{enumerate}
\item  $K_2(\varepsilon,0),K_2(0,\varepsilon)>0$, $K_2(t,u)<K(\varepsilon,0)$ for all $(t,u)$ satisfying $|t|>\varepsilon$, and $K_2(t,u)<K(0,\varepsilon)$ for all $(t,u)$ satisfying $|u|>\varepsilon$ .

\item $K_{2}^{(2,0)}(t,u),K_{2}^{(0,2)}(t,u)<-\beta$ for all $(t,u)$ satisfying $|t|,|u| \leq \epsilon$.

\end{enumerate}

 \end{enumerate} 
 \end{definition}

\begin{theorem}
\label{th:2d}
For any bivariate admissible kernel $K_2$, there exists $\nu>0$  such that for any $\sigma_{t},\sigma_{u}>0$ and a signal of the form (\ref{eq:2d_signal}) with delays $T_2:=\{t_{m},u_m\}$ satisfying the separation condition of Definition \ref{def:separation_2d}, $\left\{
c_{m}\right\}  $ and  $\{t_{m},u_{m}\}$  \textup{\emph{are uniquely defined
by the solution of}}
\[
\min_{\tilde{x}\in\mathcal{M}(\mathbb{R}^2)}\left\Vert \tilde{x}\right\Vert _{TV}\quad \mbox{subject to} \quad y(t,u)=\int_{\mathbb{R}^{2}}K_{2,\sigma}\left(t-s_{1},u-s_{2}\right)d\tilde{x}\left(s_{1},s_{2}\right),
\]
where $\mathcal{M}(\mathbb{R}^2)$ is the space of signed Borel measures
on $\mathbb{R}^{2}$.
\end{theorem}

As in the univariate case, the bivariate separation constant $\nu$ depends on the parameters of the bivariate kernel $K_2$. 
Again, flatness of $K_2$ at the origin implies the need for greater separation (see e.g. equations (\ref{cond-nu-biv}), (\ref{cond-nu-biv-2}),(\ref{cond-nu-biv-3})).

In practice, the measured signal is contaminated by noise and does not fit
exactly the above models. In this case, without a separation condition, the decomposition can not be stable by any method. To see this, consider some constants 
$t_\varepsilon,t_{0}>0$ and a signal of the form $y(t)=g(t_{0}-t)-g(t_{0}+t_\varepsilon-t)$,
where $g(t)$ is a Gaussian kernel. Then,
\begin{align*}
y(t)  & =e^{-\frac{1}{2}\left(  \frac{t_{0}-t}{\sigma}\right)  ^{2}}%
-e^{-\frac{1}{2}\left(  \frac{t_{0}+t_\varepsilon-t}{\sigma}\right)  ^{2}%
}=e^{-\frac{t^{2}}{2\sigma^{2}}}\left(  e^{-\frac{t_{0}^{2}}{2\sigma^{2}}%
}e^{\frac{t_{0}t}{\sigma^{2}}}-e^{-\frac{\left(  t_{0}+t_\varepsilon\right)  ^{2}%
}{2\sigma^{2}}}e^{\frac{\left(  t_{0}+t_\varepsilon\right)  t}{\sigma^{2}}}\right)
\\
& \leq e^{-\frac{t^{2}}{2\sigma^{2}}}e^{\frac{\left(  t_{0}+t_\varepsilon\right)
t}{\sigma^{2}}}\left(  e^{-\frac{t_{0}^{2}}{2\sigma^{2}}}-e^{-\frac{\left(
t_{0}+t_\varepsilon\right)  ^{2}}{2\sigma^{2}}}\right)  .
\end{align*}
Clearly, as $t_\varepsilon\rightarrow0$, $y(t)$ decays rapidly to zero for any
$t_{0}$ and $t$. Thus, even if the signal is contaminated with a minuscule
amount of noise or model error, there is no hope to recover $\left\{
t_{m}\right\}  $ and $\left\{  c_{m}\right\}  $. 

As aforementioned in Section \ref{Section1}, there is no tractable algorithm solving (\ref{4}).  Therefore, in addressing the noisy case we consider the sampled version of the problem in which the TV norm reduces to $\ell_1$ norm. For convenience, we focus here on the univariate model, however a similar result holds in the bivariate case. 

Let us assume the sampling interval to be $1/N$ for a given integer $N$, and that the  delays $T$ 
lie on the grid $k/N$, $k\in\mathbb{Z}$, i.e. $t_m=k_m/N$ for some $k_m\in\mathbb{Z}$.
Then, with $K_\sigma[k]:= K_\sigma(k/N)$, we obtain a discrete form of the stream of pulses
\[
y[k]:=y\left({\frac{k}{N}}\right)=\sum_m {c_m K_\sigma\left({\frac{k-k_m}{N}}\right)} = \sum_m {c_m K_\sigma[k-k_m]}.
\]
The discrete noisy model we consider is given by 
\begin{equation}
y[k]  =\left(  K_{\sigma}\ast\left(  x+n\right)  \right) [k]  ,\quad\left\Vert K_{\sigma}\ast n\right\Vert _{1}\leq\delta
,\label{5}%
\end{equation}
where $'\ast '$ denotes a discrete convolution, $x[k]:=x(k/N)$ is the underlying true superposition of delays, and $n:=\{n[k]\}$ is an additive noise or model mismatch. The discrete system can be presented in a matrix notation as 
\begin{equation*}
y=\mathbf{K}x+n,
\end{equation*} 
where $\mathbf{K}$ is the convolution matrix. This matrix is guaranteed to be invertible in many cases, such as for the Gaussian kernel, and hence one may consider estimating $x$ by solving the linear system of equations.  However, as the stability of the linear system depends linearly on the condition number of $\mathbf{K}$ (see for instance Section 4.3 in \cite{beck2014introduction}), this method will be stable only if the sampling step is large and thus will suffer from severe restrictions on the attainable  resolution. Figure \ref{fig:LS} shows the recovery of a signal from (\ref{5}) using $\ell_1$ minimization and least-squares (LS) in a noise-free setting with Gaussian kernel with standard deviation of $\sigma=0.1$ and sampling step of $0.01$. As can be seen, while the $\ell_1$ minimization perfectly recovers the signal according to Theorem \ref{Theorem1}, the LS approach  fails totally. We emphasize that the experiment was performed in a noise-free setting (i.e. $n=0$) and therefore the recovery failure of the LS is due to the amplification of the computer numerical errors. We further present the condition number of the convolution matrix $\mathbf{K}$ in the case of Gaussian kernel as a function of the discretization step. As can be seen, the condition number grows exponentially as the sampling interval gets smaller.

\begin{figure}[h!]
\begin{center}$
\begin{array}{cc}
\includegraphics[scale=0.4]{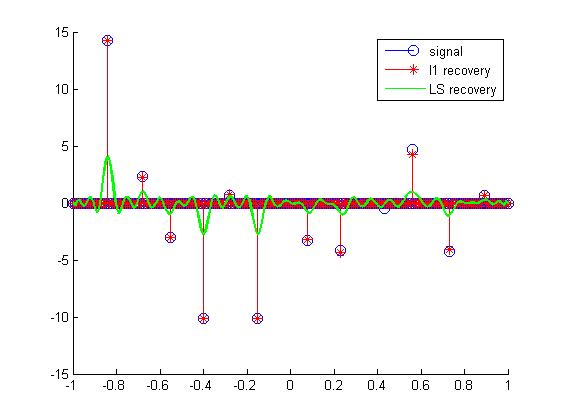}&
\includegraphics[scale=0.4]{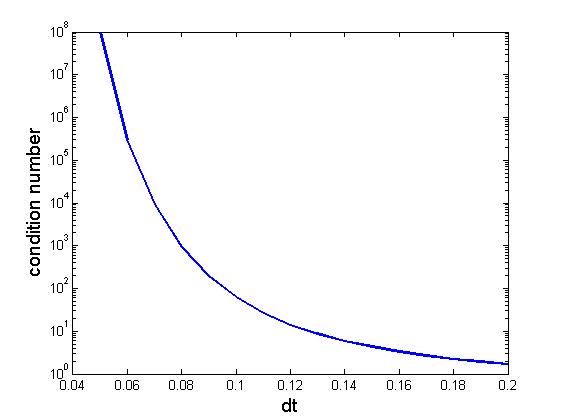} \\
\mbox{(a)} & \mbox{(b)}
\end{array}$
\end{center}
\protect\caption{\label{fig:LS} (a) Recovery of a signal from its convolution with Gaussian kernel with $\sigma=0.1$ using $\ell_1$ minimization and least-squares in a noise-free setting. The signal has a separation constant of $\nu=1.2$. The discretization step is 0.01 and the associated condition number of the convolution matrix in this case is $3.4697\cdot 10^{18}$. (b) The condition number of the convolution matrix of Gaussian kernel with $\sigma=0.1$ as a function of the discretization step $dt$ on the interval $[-1,1]$.}
\end{figure}

We suggest to estimate $x$ by the relaxed $\ell_{1}$ program
\begin{equation}
min_{\tilde{x}}\left\Vert \tilde{x}\right\Vert _{1}\quad \mbox{subject to} \quad\left\Vert
y-K_{\sigma}\ast\tilde{x}\right\Vert _{1}\leq\delta.\label{6}%
\end{equation}
The following result shows that the error is proportional to the noise level
$\delta$.

\begin{theorem}
\label{Theorem2} Consider the model \emph{(\ref{5}) } for an admissible kernel $K$. If $T$ satisfies
the separation condition of Definition \ref{Definition1} for  $\sigma >0$, then the  solution
$\hat{x}$ of {(\ref{6}) }obeys

\begin{equation*}
\left\Vert \hat{x} - x \right\Vert _{1}\leq\frac{72K(0)\left\vert K^{(2)}(0)\right\vert\gamma^2}{9\beta K(0)\left\vert K^{(2)}(0)\right\vert-D_1\nu^{-2}-D_2\nu^{-4}}\delta,
\end{equation*}
where 
\begin{eqnarray*}
\gamma & := & \max \{ N\sigma , \varepsilon^{-1} \}, \\
D_{1} & := & 3\pi^2 \left(C_2\beta +2C_0\beta  +8C_1^2K(0)\gamma^2 \right) ,\\
D_{2} & := &4\pi^4C_1^2K(0)\gamma^2 .
\end{eqnarray*}
Therefore, for sufficiently large $\nu$, we obtain 
\[
\left\Vert  \hat{x} - x \right\Vert _{1} \le \frac{16\gamma^2}{\beta }\delta.
\]
\end{theorem}

As can be seen, high value of $\beta$ results in small recovery error. So kernels which are flat near the origin will be less stable in a noisy environment.

In a consecutive paper, it was shown that the solution of (\ref{6}) is also clustered around the true support \cite{bendory2015stable}. The case of non-negative stream of pulses, i.e. $c_m>0$, was analyzed in \cite{bendory2015robust}. In this work it was proven that the separation is unnecessary in this case and can be replaced by the notion of Rayleigh regularity.  
 
\section{Proof of main results\label{Section3}}

The main pillar of the forthcoming proofs is a duality theorem which is a variant of the `dual certificate' theorems of \cite{bendory2013exact}, \cite{bendory2013Legendre} and \cite{candes2013towards}.

\subsection{The duality theorem}

\begin{theorem}
\label{Theorem3}Let $x(t)=\sum_{m}c_{m}\delta_{t_{m}}(t)$, \textup{$c_{m}%
\in\mathbb{R}$} where $T:=\{t_{m}\}\subseteq\mathbb{R}$, and let
$y(t)=\int_{\mathbb{R}}K(t-s)dx\left(  s\right)  $ for an even, $L$
times differentiable kernel $K(t)$. If for any set $\{v_{m}\}$, $v_{m}%
\in\mathbb{R} $, with $|v_{m}|=1$, there exists a function of the form
\begin{equation}
q(t)=\int_{\mathbb{R}}\sum_{\ell=0}^{L}K^{(\ell)}(t-s)d\mu_{\ell}\left(
s\right)  ,\label{7}%
\end{equation}
for some measures \textup{$\left\{  \mu_{\ell}\left(  t\right)  \right\}
_{\ell=0}^{L}$}, satisfying
\begin{align}
q(t_{m}) &  =v_{m}\,,\,\forall t_{m}\in T,\label{8}\\
|q(t)| &  <1\,,\,\forall t\in\mathbb{R}\backslash T,\label{9}%
\end{align}
then $x$ is the unique real Borel measure solving
\begin{equation}
\min_{\tilde{x}\in\mathcal{M\left(  \mathbb{R}\right)  }}\Vert\tilde{x}%
\Vert_{TV}\quad\mbox{subject to}\quad y(t)=\int_{\mathbb{R}}K(t-s)d\tilde
{x}\left(  s\right)  .\label{10}%
\end{equation}

\end{theorem}

\begin{proof}
Let $\hat{x}$ be a solution of (\ref{10}), and define $\hat{x}=x+h$. The
difference measure $h$ can be decomposed relative to $|x|$ as
\[
h=h_{T}+h_{T^{C}},
\]
where $h_{T}$ is supported in $T$, and $h_{T^{C}}$ is supported in
$T^{C}$ (the complementary of $T$). 
If $h_T=0$, then also $h_{T^c}=h = 0$. Otherwise, $\Vert\hat{x}
\Vert_{TV}>\Vert{x}
\Vert_{TV}$ which is a contradiction. If $h_T\neq 0$, we 
perform a polar decomposition of
$h_{T}$ 
\[
h_{T}=|h_{T}|sgn(h_{T}),
\]
where $sgn(z):=\frac{z}{\vert z \vert}$. By assumption, for any $0\leq\ell\leq L$
\[
\int_{\mathbb{R}}K^{(\ell)}(t-s)d\hat{x}\left(  s\right)  =\int_{\mathbb{R}%
}K^{(\ell)}(t-s)dx\left(  s\right)  ,
\]
which in turn leads to $\int_{\mathbb{R}}K^{(\ell)}(t-s)dh\left(  s\right)
=0$. Then, for any $q$ of the form (\ref{7}), since $K$ is even, we get
\begin{align*}
\left\langle q,h\right\rangle  &: =\int_{\mathbb{R}}q(t)dh\left(  t\right)  \\
& =\int_{\mathbb{R}}\left(  \int_{\mathbb{R}}\sum_{\ell=0}^{L}K^{(\ell
)}(t-s)d\mu_{\ell}\left(  s\right)  \right)  dh\left(  t\right)  \\
& =\int_{\mathbb{R}}\sum_{\ell=0}^{L}d\mu_{\ell}\left(  s\right)
{\int_{\mathbb{R}}K^{(\ell)}(t-s)dh\left(  t\right)  }\\
& =\int_{\mathbb{R}}\sum_{\ell=0}^{L}d\mu_{\ell}\left(  s\right)(-1)^{\ell}
\underbrace{\int_{\mathbb{R}}K^{(\ell)}(s-t)dh\left(  t\right)  }_{0}\\
& =0.
\end{align*}
By assumption, for the choice $v_{m}=$ $sgn(h_{T}\left(  t_{m}\right)  )$, there exists $q$ of the form (\ref{7}),
such that
\begin{align*}
q(t_{m}) &  =sgn(h_{T}(t_{m}))\,,\,\forall t_{m}\in T,\nonumber\\
|q(t)| &  <1\,,\,\forall t\in\mathbb{R}\backslash T.%
\end{align*}
Consequently,
\[
0=\left\langle q,h\right\rangle =\left\langle q,h_{T}\right\rangle
+\left\langle q,h_{T^{C}}\right\rangle =\Vert h_{T}\Vert_{TV}+\left\langle
q,h_{T^{C}}\right\rangle .
\]
If $h_{T^{C}}=0$, then $\Vert h_{T}\Vert_{TV}=0$, and $h=0$. Alternatively, if
$h_{T^{C}}\neq0$, from the second property of $q$,
\[
|\left\langle q,h_{T^{C}}\right\rangle |<\Vert h_{T^{C}}\Vert_{TV}.
\]
Thus,
\[
\Vert h_{T^{C}}\Vert_{TV}>\Vert h_{T}\Vert_{TV}.
\]
As a result, using the fact that $\hat{x}$ has minimal TV norm, we get
\[%
\begin{split}
\Vert x\Vert_{TV} &  \geq\Vert x+h\Vert_{TV}=\Vert x+h_{T}\Vert_{TV}+\Vert
h_{T^{C}}\Vert_{TV}\\
&  \geq\Vert x\Vert_{TV}-\Vert h_{T}\Vert_{TV}+\Vert h_{T^{C}}\Vert_{TV}>\Vert
x\Vert_{TV},
\end{split}
\]
which is a contradiction. Therefore, $h=0$, which implies that $x$ is the
unique solution of (\ref{10}).
\end{proof}

\subsection{Proof of Theorem \ref{Theorem1}} \label{sec:proof1}

For simplicity and without loss of generality we will assume throughout the proof  that $\sigma=1$. 
To prove Theorem \ref{Theorem1} we make use of the following result:

\begin{proposition}
\label{proposition1} Let $T$ satisfy the separation condition of Definition
\ref{Definition1} and let $\left\{
v_{m}\right\}  $ be any set as in Theorem \ref{Theorem3}. Then, there exist
coefficients $\left\{  a_{m}\right\}  $ and $\left\{  b_{m}\right\}  $ such
that%
\begin{equation}
q\left(  t\right)  =\sum_{m}\left(  a_{m}K\left(  t-t_{m}\right)
+b_{m}K^{\left(  1\right)  }\left(  t-t_{m}\right)  \right), \label{11}%
\end{equation}
satisfies:
\begin{align}
q(t_{m})&=v_{m}\,,\,\forall t_{m}\in T, \label{11a} \\
q^{\left(  1\right)  }\left(  t_{m}\right)  &=0\,,\,\forall t_{m}\in
T.\label{12}%
\end{align}
Furthermore, The coefficients can be bounded by
\begin{align}
\left\Vert \mathbf{a}\right\Vert _{\infty}  & :=\underset{m}{\max}\left\vert
a_{m}\right\vert \leq\frac{3\nu^2}{ 3K\left(  0\right)\nu^2
-2\pi^2C_0},\label{23}\\
\left\Vert \mathbf{b}\right\Vert _{\infty}  & :=\underset{m}{\max}\left\vert
b_{m}\right\vert \leq\frac{ \pi^2C_1}{\left( 3 \left\vert
K^{\left(  2\right)  }\left(  0\right)  \right\vert\nu^2 -\pi^2C_2  \right)  \left(   3K\left(  0\right)\nu^2
-2\pi^2C_0\right)  },\label{24} 
\end{align}
where $\nu$ is the separation constant from Definition \ref{Definition1}.
If $v_m=1$, we also have
\begin{equation}
a_{m}\geq\frac{1}{K\left(  0\right)  }\left(  1-\frac{2\pi^2C_0 }{ 3K\left(  0\right)\nu^2
-2\pi^2C_0}\right). \label{13}%
\end{equation}

\end{proposition}

Proposition \ref{proposition1} suggests a candidate $q$ to use in Theorem \ref{Theorem3}
and once proved, it will guarantee $q$ satisfies (\ref{8}). The next two results are needed to prove
that $q$ as in (\ref{11}), satisfies (\ref{9}) as well
so as to complete the proof of Theorem \ref{Theorem1}. 

\begin{lemma} \label{lemma1}
Under the separation condition of Definition \ref{Definition2} with $\varepsilon<\nu/2$, $q$ as in Proposition \ref{proposition1} satisfies $\vert q(t)\vert<1$ for all $t$ obeying  
$0<\vert t-t_m\vert \le \varepsilon$ for some $t_m \in T$. 
\end{lemma}

\begin{lemma} \label{lemma2}
Under the separation condition of Definition \ref{Definition2} with $\varepsilon<\nu/2$, $q$ as in Proposition \ref{proposition1} satisfies $\vert q(t)\vert<1$
for all $t$ obeying $\vert t-t_m\vert>\varepsilon$ for all $t_m\in T$. 
\end{lemma}

\subsubsection{Proof of Proposition \ref{proposition1}}

Substituting the requirements (\ref{11a}) and (\ref{12}) we get the
set of equations%
\begin{align*}
\sum_{m}\left(  a_{m}K\left(  t_{k}-t_{m}\right)  +b_{m}K^{\left(  1\right)
}\left(  t_{k}-t_{m}\right)  \right)   & =v_{k},\\
\sum_{m}\left(  a_{m}K^{\left(  1\right)  }\left(  t_{k}-t_{m}\right)
+b_{m}K^{\left(  2\right)  }\left(  t_{k}-t_{m}\right)  \right)   & =0,
\end{align*}
for all $t_{k}\in T$, which can be written in a matrix vector form as
\begin{equation}
\left[
\begin{array}
[c]{cc}%
G_{0} & G_{1}\\
G_{1} & G_{2}%
\end{array}
\right]  \left[
\begin{array}
[c]{c}%
\mathbf{a}\\
\mathbf{b}%
\end{array}
\right]  =\left[
\begin{array}
[c]{c}%
\mathbf{v}\\
\mathbf{0}%
\end{array}
\right] ,\label{14}%
\end{equation}
where $\left(  G_{\ell}\right)  _{k,m}:=K^{(\ell)}\left(  t_{k}-t_{m}\right)  $, $\ell=0,1,2,$ and $\mathbf{a}:=\{a_m\}, \mathbf{b}:=\{b_m\}, \mathbf{v}:=\{v_m\}$. From standard linear algebra (see e.g. \cite{zhang2006schur}) we know that the matrix in (\ref{14}) is invertible if both
$G_{2}$ and its Schur complement $S:=G_{0}-G_{1}G_{2}^{-1}G_{1}$ are invertible. Also, recall that a
matrix $M$ is invertible if there exists   $\alpha\neq0$ such that  $\left\Vert \alpha I-M\right\Vert _{\infty
}<\left\vert \alpha\right\vert $, where $\left\Vert
M\right\Vert _{\infty}:=\underset{i}{\max}\sum_{j}\left\vert m_{i,j}\right\vert$. In such a case we have, 
\begin{equation}
\left\Vert M^{-1}\right\Vert _{\infty}\leq\frac{1}{\left\vert \alpha
\right\vert -\left\Vert \alpha I-M\right\Vert _{\infty}},\label{15}%
\end{equation}
(see for instance Corollary 5.6.16 in \cite{hornmatrix}).
Using the properties of an admissible kernel (see Definition
\ref{Definition2}) and the separation condition
(Definition \ref{Definition1}), we can estimate
\begin{align*}
\left\Vert K^{\left(  2\right)  }\left(  0\right)  I-G_{2}\right\Vert
_{\infty}  & =\underset{k}{\max}\sum_{m\neq k}\left\vert K^{\left(  2\right)
}\left(  t_{k}-t_{m}\right)  \right\vert \nonumber\\
& \leq C_{2}\underset{k}{\max}\sum_{m\neq k}\frac{1}{1+\left(  t_{k}%
-t_{m}\right)  ^{2}}\nonumber\\
& \leq C_{2}\underset{k}{\max}\sum_{m\neq k}\frac{1}{1+\left(  \left(
k-m\right)  \nu\right)  ^{2}}.\label{16}%
\end{align*}
It can be shown that
\begin{equation}
\sum_{n=1}^{\infty}\frac{1}{1+\left(  n\nu\right)  ^{2}} < \frac{\pi^{2}%
}{6\nu^{2}}:=E\left(  \nu\right). \label{eq:E_def}%
\end{equation}
So, we readily get%
\begin{equation}
\left\Vert K^{\left(  2\right)  }\left(  0\right)  I-G_{2}\right\Vert
_{\infty}\leq2C_{2}E\left(  \nu\right).\label{18}%
\end{equation}
Therefore, $G_{2}$ is invertible if $2C_{2}E(\nu) < \left|{K^{(2)}(0)} \right|$, which is equivalent to the condition
\begin{equation} \label{estimate-nu-1}
\nu^2 > \frac{C_2 \pi^2}{3\left|{K^{(2)}(0)} \right|}.
\end{equation}
\begin{remark} \label{remark5}
 We see here that if $K$ is relatively flat at the origin, a larger separation is required for unique recovery
through TV minimization.
 \end{remark}
\noindent Next we consider
\begin{align}
\left\Vert K\left(  0\right)  I-S\right\Vert _{\infty}  & =\left\Vert K\left(
0\right)  I-G_{0}+G_{1}G_{2}^{-1}G_{1}\right\Vert _{\infty}\nonumber\\
& \leq\left\Vert K\left(  0\right)  I-G_{0}\right\Vert _{\infty}+\left\Vert
G_{1}\right\Vert _{\infty}^{2}\left\Vert G_{2}^{-1}\right\Vert _{\infty
}.\label{19}%
\end{align}
Using the same method leading to (\ref{18}), we readily observe that%
\begin{equation*}
\left\Vert K\left(  0\right)  I-G_{0}\right\Vert _{\infty}\leq2C_{0}E\left(
\nu\right), \label{20}%
\end{equation*}
and since $K^{(1)}(0)=0$ we also have
\begin{equation}
\left\Vert G_{1}\right\Vert _{\infty}\leq2C_{1}E\left(  \nu\right). \label{21}%
\end{equation}
Furthermore, using (\ref{15}) and (\ref{18}) we get%
\begin{align}
\left\Vert G_{2}^{-1}\right\Vert _{\infty}  & \leq\frac{1}{\left\vert
K^{\left(  2\right)  }\left(  0\right)  \right\vert -\left\Vert K^{\left(
2\right)  }\left(  0\right)  I-G_{2}\right\Vert _{\infty}}\nonumber\\
& \leq\frac{1}{\left\vert K^{\left(  2\right)  }\left(  0\right)  \right\vert
-2C_{2}E\left(  \nu\right)  }\label{21a}.
\end{align}
Substitution in (\ref{19}) results in
\begin{align}
\left\Vert K\left(  0\right)  I-S\right\Vert _{\infty}  & \leq2C_{0}E\left(
\nu\right)  +\frac{\left(  2C_{1}E\left(  \nu\right)  \right)  ^{2}%
}{\left\vert K^{\left(  2\right)  }\left(  0\right)  \right\vert
-2C_{2}E\left(  \nu\right)  }\nonumber\\
& = 2C_{0} E\left(\nu\right)  \left(  1+\frac{2C_{1}^2 E\left(  \nu\right) }{C_{0}\left(  \left\vert K^{\left(  2\right)  }\left(
0\right)  \right\vert -2C_{2}E\left(  \nu\right)  \right)  }\right)
\nonumber\\
& \leq 4 C_{0} E\left(  \nu\right),
\label{21b}%
\end{align}
where the last inequality holds for 
\begin{equation} \label{estimate-nu-2}
\nu^2 \ge \frac{\pi^2(C_1^2+C_0C_2)}{3C_0\left\vert K^{(2)}(0)\right\vert}.
\end{equation}
Therefore, if further
\begin{equation} \label{estimate-nu-3}
\nu^2 > \frac{2\pi^2C_0}{3K(0)},
\end{equation}
then $\left\Vert K\left(  0\right)  I-S\right\Vert _{\infty} \le 4C_0 E(\nu)< K\left(  0\right)$ and $S$ is invertible.
Hence, (\ref{14}) has a unique solution. 
Furthermore, we conclude that $\mathbf{a}$ and
$\mathbf{b}$ are given by%
\[
\left[
\begin{array}
[c]{c}%
\mathbf{a}\\
\mathbf{b}%
\end{array}
\right]  =\left[
\begin{array}
[c]{cc}%
G_{0} & G_{1}\\
G_{1} & G_{2}%
\end{array}
\right]  ^{-1}\left[
\begin{array}
[c]{c}%
\mathbf{v}\\
\mathbf{0}%
\end{array}
\right],
\]
implying 
\begin{equation}
\left[
\begin{array}
[c]{c}%
\mathbf{a}\\
\mathbf{b}%
\end{array}
\right]  =\left[
\begin{array}
[c]{c}%
S^{-1}\mathbf{v}\\
-G_{2}^{-1}G_{1}S^{-1}\mathbf{v}%
\end{array}
\right] .\label{22}%
\end{equation}
Hence, by (\ref{15}) and (\ref{21b}) we get
\begin{align*} 
\left\Vert \mathbf{a}\right\Vert _{\infty}  & \leq\left\Vert S^{-1}\right\Vert
_{\infty}\nonumber\\
& \leq\frac{1}{ K\left(  0\right) -\left\Vert K\left(
0\right)  I-S\right\Vert _{\infty}}\nonumber\\
& \leq\frac{1}{ K\left(  0\right) -4E\left(
\nu\right)  C_{0}}.%
\end{align*} 
Using (\ref{21}) and (\ref{21a}) we also have
\begin{align*}
\left\Vert \mathbf{b}\right\Vert _{\infty}  & \leq\left\Vert G_{2}%
^{-1}\right\Vert _{\infty}\left\Vert G_{1}\right\Vert _{\infty}\left\Vert
S^{-1}\right\Vert _{\infty}\nonumber\\
& \leq\frac{2C_{1}E\left(  \nu\right)  }{\left(  \left\vert K^{\left(
2\right)  }\left(  0\right)  \right\vert -2C_{2}E\left(  \nu\right)  \right)
\left(   K\left(  0\right)  -4E\left(  \nu\right)
C_{0}\right)  }.%
\end{align*}
Assuming $v_k=1$ we get
\begin{align*}
a_{k}  & =\left(  S^{-1}\mathbf{v}\right)  _{k}=\frac{1}{K\left(  0\right)  }\left(  1-\left(  S^{-1}\left(  S-K\left(  0\right)
I\right)  \mathbf{v}\right)  _{k}\right) \nonumber \\
& \geq\frac{1}{K\left(  0\right)  }\left(  1-\left\Vert S^{-1}\right\Vert
_{\infty}\left\Vert S-K\left(  0\right)  I\right\Vert _{\infty}\right), 
\end{align*}
and using  (\ref{21b}) we end up with%
\begin{equation}
a_{k}\geq\frac{1}{K\left(  0\right)  }\left(  1-\frac{4C_{0}E\left(
\nu\right)  }{K\left(  0\right)  -4C_{0}E\left(  \nu\right)  }\right).
 \label{eq:ak1}
\end{equation}
This completes the proof.

\subsubsection{Proof of Lemma \ref{lemma1}}

Assume without loss of generality that $t\in\mathbb{R}$, where $t_k< t\leq t_k+\varepsilon$, for some $t_k\in T$
and that $q(t_k)=v_k=1$. The proof is similar for the case $t_k-\epsilon \le t< t_k$ or $v_k=-1$.
Since $\left| {t-t_{m} } \right|>\nu /2$ for $m\ne k$, we have, using the separation assumption, for $\ell=0,1,2,3$, 
\begin{align*}
 \sum\limits_{m\ne k} {\left| {K^{\left( \ell \right)}\left( {t-t_{m} } 
\right)} \right|} & \le C_{\ell} \sum\limits_{m\ne k} {\frac{1}{1+\left( {t-t_{m} 
} \right)^{2}}} \\ 
& \le C_{\ell} \sum\limits_{m\ne k} {\frac{1}{1+\left( {2^{-1}\left( {k-m} 
\right)\nu } \right)^{2}}} \\ 
& \le 8C_{\ell} E\left( \nu \right). \\ 
\end{align*}
Using this estimate, as well as (\ref{23}), (\ref{24}) and (\ref{13}) we obtain 
\begin{align*}
q^{\left( 2 \right)}\left( t \right) &=\sum\limits_m {a_{m} K^{\left( 2 \right)}\left( {t-t_{m} }
\right)+b_{m} K^{\left( 3 \right)}\left( {t-t_{m} } \right)} \\ 
& \le a_{k} K^{\left( 2 \right)}\left( {t-t_{k} } \right)+\left\| a 
\right\|_{\infty } \sum\limits_{m\ne k} {\left| {K^{\left( 2 \right)}\left( 
{t-t_{m} } \right)} \right|} +\left\| b \right\|_{\infty } \sum\limits_{m} {\left| {K^{\left( 3 \right)}\left( {t-t_{m} } \right)} \right|} \\ 
& \le -\beta \frac{1}{K\left( 0 \right)}\left( {1-\frac{2\pi^{2}C_{0} 
}{3K\left( 0 \right)\nu^{2}-2\pi^{2}C_{0} }} \right)+\frac{4\pi^{2}C_{2} 
}{3K\left( 0 \right)\nu^{2}-2\pi^{2}C_{0} } \\
& + \frac{ \pi^2C_1C_3}{\left( 3 \left\vert
K^{\left(  2\right)  }\left(  0\right)  \right\vert\nu^2 -\pi^2C_2  \right)  \left(   3K\left(  0\right)\nu^2
-2\pi^2C_0\right)  }\left(1+\frac{4\pi^2}{3\nu^{2}}\right). \\ 
\end{align*}
Thus, it can be shown that $q^{(2)}(t)< -\frac{\beta}{2K(0)}$ for sufficiently large $\nu$ that depends on the 
parameters of $K$.

By Taylor Remainder theorem, for any $t_k<t<t_k+\varepsilon$, there exists $t_k<\xi\leq  t$ such that 

\begin{equation} \label{eq:taylor}
q(t)=q(t_k)+q^{(1)}(t_k)(t-t_k)+\frac{1}{2}q^{(2)}(\xi)(t-t_k)^2.
\end{equation}
Since by construction $q^{(1)}(t_k)=0$, we conclude that for sufficently large $\nu$
\begin{equation}\label{eq:q_up}
q(t) \leq 1-\frac{\beta}{4K(0)}(t-t_k)^2, 
\end{equation}
implying that $q(t)<1$, for $t_k < t \le t_k+\epsilon$.
 
To complete the proof we need to show also that $q(t)>-1$. We then use again the properties of the kernel,
(\ref{23}), (\ref{24}) and (\ref{13}) to estimate

\begin{align*}
q(t) &=\sum\limits_m {a_{m} K(t-t_{m})+b_{m} K^{(1)}(t-t_{m})} \\ 
& \ge a_{k} K(t-t_{k})-\left\| a 
\right\|_{\infty } \sum\limits_{m\ne k} {\left| {K\left( 
{t-t_{m} } \right)} \right|}-\left\| b \right\|_{\infty } \sum\limits_{m} {\left| {K^{\left( 1 \right)}\left( {t-t_{m} } \right)} \right|} \\ 
& \ge \frac{K(\epsilon)}{K\left( 0 \right)}\left( {1-\frac{2\pi^{2}C_{0} 
}{3K\left( 0 \right)\nu^{2}-2\pi^{2}C_{0} }} \right)-\frac{4\pi^{2}C_{0} 
}{3K\left( 0 \right)\nu^{2}-2\pi^{2}C_{0} } \\
& - \frac{ \pi^2C_1^2}{\left( 3 \left\vert
K^{\left(  2\right)  }\left(  0\right)  \right\vert\nu^2 -\pi^2C_2  \right)  \left(   3K\left(  0\right)\nu^2
-2\pi^2C_0\right)  }\left(1+\frac{4\pi^2}{3\nu^{2}}\right). \\ 
\end{align*}
This implies that for sufficiently large $\nu$, we have $q(t)>-1$, for $t_k \le t \le t_k+\epsilon$. We therefore conclude
\begin{equation*}
\vert q(t)\vert <1,\quad \forall \thinspace 0< \vert t-t_k\vert \le \varepsilon, \thinspace t_k\in T.
\end{equation*}

\subsubsection{Proof of Lemma \ref{lemma2}}

Fix $t\in\mathbb{R}$ satisfying $\vert t-t_m\vert >\varepsilon$ for all $t_m\in T$, and denote $t_k:=\min\{\vert t_m-t\vert \thinspace:\thinspace t_m\in T\}$. This implies that $|t-t_m| > \nu /2$, for all $m \ne k$. Then, from (\ref{11}), the properties of admissible kernel, (\ref{23}) and (\ref{24}), 
\begin{align*}
 \left| {q\left( t \right)} \right| &\le \left\| a \right\|_{\infty } \left( 
{\left| {K\left( {t-t_{k} } \right)} \right|+\sum\limits_{m\ne k} {\left| 
{K\left( {t-t_{k} } \right)} \right|} } \right)+\left\| b \right\|_{\infty } 
\left( {\left| {K^{\left( 1 \right)}\left( {t-t_{k} } \right)} 
\right|+\sum\limits_{m\ne k} {\left| {K^{\left( 1 \right)}\left( {t-t_{k} } 
\right)} \right|} } \right) \\ 
& \le \left\| a \right\|_{\infty } \left( {K\left( \varepsilon \right)+8C_{0} 
E\left( \nu \right)} \right)+\left\| b \right\|_{\infty } \left( 
{\frac{C_{1} }{1+\varepsilon^{2}}+8C_{1} E\left( \nu \right)} \right) \\ 
& \le \frac{3K\left( \varepsilon \right)\nu^2}{3K\left( 0 \right)\nu^2-2\pi^{2}C_{0} 
}+\frac{4\pi^{2}C_{0} }{3K\left( 0 \right)\nu^{2}-2\pi^{2}C_{0}} \\
& + \frac{ \pi^2C_1^2}{\left( 3 \left\vert
K^{\left(  2\right)  }\left(  0\right)  \right\vert\nu^2 -\pi^2C_2  \right)  \left(   3K\left(  0\right)\nu^2
-2\pi^2C_0\right)  }\left( {\frac{1 
}{1+\varepsilon^{2}}+\frac{4\pi^{2} }{3\nu^{2}}} \right). \\ 
\end{align*}
By the Taylor Remainder theorem and the properties of $K(t)$, one has $ 0< K(\varepsilon) \le K(0) - \beta \varepsilon^2 /2$,
which yields
\begin{equation*} 
0  < \frac{K(\varepsilon)}{K(0)} \le 1 - \frac{\beta \varepsilon^2}{2K(0)}.
\end{equation*}
Therefore, it is obvious that for sufficiently 
large $\nu$, we get that 
\begin{equation} \label{eq:q_far}
\left\vert q\left(  t\right)  \right\vert <1 - \frac{\beta \varepsilon^2}{4K(0)}.
\end{equation}
This completes the proof.

\subsection{Proof of Theorem \ref{th:2d}}

For simplicity and without loss of generality we assume that $\sigma_{t}=\sigma_{u}=1$. 
 The proof follows the outline of the proof in the univariate case. We make use of the following result:
\begin{proposition}
\label{prop_2d-1}Let $T_{2}$ satisfies the separation condition
of Definition \ref{def:separation_2d} for the bivariate admissible kernel
in Definition \ref{def:2d_kernel}, and let $\left\{ v_{m}\right\} $
be any set as in Theorem \ref{Theorem3}. Then, there exist coefficients $\left\{ a_{m}\right\} ,\left\{ b_{m}\right\} $
and $\left\{ c_{m}\right\} $ such that

\begin{equation}
q\left(t,u\right)=\sum_{m}\left(a_{m}K_{2}\left(t-t_{m},u-u_{m}\right)+b_{m}K_{2}^{\left(1,0\right)}\left(t-t_{m},u-u_{m}\right)+c_{m}K_{2}^{\left(0,1\right)}\left(t-t_{m},u-u_{m}\right)\right),\label{eq:q2def}
\end{equation}
satisfies for all $\left(t_{m},u_{m}\right)\in T_{2}$:
\begin{align}
q\left(t_{m},u_{m}\right) & = v_{m},\label{eq:q2} \\
q^{(1,0)}\left(t_{m},u_{m}\right) & =q^{(0,1)}\left(t_{m},u_{m}\right)=  0. \label{eq:q2_grad}
\end{align}
The coefficients are bounded by 

\begin{align}
\left\Vert \mathbf{a}\right\Vert _{\infty} & :=\max_{m}\left\vert a_{m}\right\vert \leq\frac{2\nu^3}{2K_{2}\left(0,0\right)\nu^3-9\pi^2C_{0,0}}, \label{eq:a}\\
\left\Vert \mathbf{b}\right\Vert _{\infty} & :=\max_{m}\left\vert b_{m}\right\vert \leq\frac{6\pi^2C_{1,0}}{\left(2\nu^3\left|K_{2}^{\left(2,0\right)}(0,0)\right|-6\pi^2C_{2,0}\right)\left(2\nu^3K_{2}\left(0,0\right)-9\pi^2C_{0,0}\right)},\label{eq:b}\\
\left\Vert \mathbf{c}\right\Vert _{\infty} & :=\max_{m}\left\vert c_{m}\right\vert \leq\frac{6\pi^2C_{0,1}}{\left(2\nu^3\left|K_{2}^{\left(0,2\right)}(0,0)\right|-3\pi^2C_{0,2}\right)\left(2\nu^3K_{2}\left(0,0\right)-9\pi^2C_{0,0}\right)}. \label{eq:c}
\end{align}

If $v_m=1$, we also have 
\begin{equation} \label{eq:ak2}
a_{m}\geq\frac{1}{K_{2}\left(0,0\right)}\left(1-\frac{9\pi^2C_{0,0}}{2\nu^3K_{2}\left(0,0\right)-9\pi^2C_{0,0}}\right).
\end{equation}
\end{proposition}

Proposition \ref{prop_2d-1} suggests a candidate $q\left(  t,u\right)  $ to
use in Theorem \ref{Theorem3}. 
Next, we define the sets
\begin{eqnarray} \label{eq:S}
S_{k} & := & \left\{ (t,u)\thinspace:\thinspace 0<\vert t-t_k\vert,\vert u-u_k\vert <\varepsilon_1\right\}, \quad (t_k,u_k)\in T_2,\\
S & := & \bigcup_{(t_k,u_k)\in T_{2}}S_{k},\nonumber  \\
 S^{C} & := & \mathbb{R}^{2}\backslash\left\{ S\cup T_2\right\},\nonumber  
\end{eqnarray}
where $\varepsilon_1\leq \varepsilon$ is a sufficiently small constant to be chosen later.
The following Lemmas complete the proof:
\begin{lemma} \label{lemma:2d_1}
Assuming the separation condition of Definition \ref{def:separation_2d} and $\varepsilon < \nu /2$, then $q(t,u)$  as in Proposition \ref{prop_2d-1} satisfies $\vert q(t,u)\vert<1$ for all $(t,u)\in S$.
\end{lemma}

\begin{lemma} \label{lemma:2d_2}
Assuming the separation condition of Definition \ref{def:separation_2d} and $\varepsilon < \nu /2$, then $q(t,u)$ as in Proposition \ref{prop_2d-1} satisfies $\vert q(t,u)\vert<1$ for all $(t,u)\in S^C$. 
\end{lemma}

\subsubsection{Proof of Proposition \ref{prop_2d-1}}

We begin the proof with a preliminary calculation. Fix $\left(t_{k},u_{k}\right)\in T_{2}$.
Let $\Omega_{n}$ be the $n^{th}$ 'rectangular ring' about $\left(t_{k},u_{k}\right)$
such that 
\[
\Omega_{n}:=\left\{ \left(t,u\right)\in\mathbb{R}^{2}\thinspace:\thinspace n\nu\le\max\left\{ \left|t-t_{k}\right|,\left|u-u_{k}\right|\right\} \le\left(n+1\right)\nu\right\} , \qquad ,n\ge1, 
\]
where $\nu$ is the separation constant from Definition \ref{def:separation_2d}.
The area of the $n^{th}$ ring is 
\[
\left|\Omega_{n}\right|=4\left(n+1\right)^{2}\nu^{2}-4n^{2}\nu^{2}.
\]
By assumption, the set $T_{2}$ satisfies the separation condition
of Definition \ref{def:separation_2d}. Hence, the points are 
centers of pairwise disjoint rectangles of area $4\nu^{2}$. 
Also, the rectangle of any $\left(t_{k},u_{k}\right)\in\Omega_{n}$
is contained in the ring 
\[
\tilde{\Omega}_{n}:=\left\{ \left(t,u\right)\in\mathbb{R}^{2}\thinspace:\thinspace\left(n-1\right)\nu\le\max\left\{ \left|t_{k}-t\right|,\left|u_{k}-u\right|\right\} \le\left(n+2\right)\nu\right\} .
\]
Therefore, we can bound the number of points of $T_2$ contained in the ring $\Omega_{n}$
by 
\begin{eqnarray}
\#\left\{ t_{k}\in\Omega_{n}\right\}  & \le & \frac{\left|\tilde{\Omega}_{n}\right|}{4\nu^{2}}\leq\frac{4\left(n+2\right)^{2}\nu^{2}-4\left(n-1\right)^{2}\nu^{2}}{4\nu^{2}}\label{eq:ring}\\
 & = & 6n+3\le9n,\nonumber 
\end{eqnarray}
for $n\geq 1$.
Equipped with (\ref{eq:ring}), we follow the outline of the proof
of Proposition \ref{proposition1}. We write (\ref{eq:q2}) and (\ref{eq:q2_grad}) explicitly 
\begin{eqnarray*}
\sum_{m}\left(a_{m}K_{2}\left(t_{k}-t_{m},u_{k}-u_{m}\right)+b_{m}K_{2}^{\left(1,0\right)}\left(t_{k}-t_{m},u_{k}-u_{m}\right)+c_{m}K_{2}^{\left(0,1\right)}\left(t_{k}-t_{m},u_{k}-u_{m}\right)\right) & = & v_{k},\\
\sum_{m}\left(a_{m}K_{2}^{\left(1,0\right)}\left(t_{k}-t_{m},u_{k}-u_{m}\right)+b_{m}K_{2}^{\left(2,0\right)}\left(t_{k}-t_{m},u_{k}-u_{m}\right)+c_{m}K_{2}^{\left(1,1\right)}\left(t_{k}-t_{m},u_{k}-u_{m}\right)\right) & = & 0,\\
\sum_{m}\left(a_{m}K_{2}^{\left(0,1\right)}\left(t_{k}-t_{m},u_{k}-u_{m}\right)+b_{m}K_{2}^{\left(1,1\right)}\left(t_{k}-t_{m},u_{k}-u_{m}\right)+c_{m}K_{2}^{\left(0,2\right)}\left(t_{k}-t_{m},u_{k}-u_{m}\right)\right) & = & 0,
\end{eqnarray*}
 for all $\left(t_{k},u_{k}\right)\in T_{2}$. This can be written
in a matrix vector form as 
\begin{equation}
\begin{bmatrix}G^{(0,0)} & G^{(1,0)} & G^{(0,1)}\\
G^{(1,0)} & G^{(2,0)}& G^{(1,1)}\\
G^{(0,1)} & G^{(1,1)} & G^{(0,2)}
\end{bmatrix}\begin{bmatrix}\mathbf{a}\\
\mathbf{b}\\
\mathbf{c}
\end{bmatrix}=\begin{bmatrix}\mathbf{v}\\
0\\
0
\end{bmatrix},\label{eq:matrix_2d}
\end{equation}
where $\mathbf{a}:=\left\{ a_{m}\right\},\mathbf{b}:=\left\{ b_{m}\right\}$,
$\mathbf{c}:=\left\{ c_{m}\right\}$, $\mathbf{v}:=\left\{ v_{m}\right\}$,
and $\left(G^{(\ell_1,\ell_2)}\right)_{k,m}  :=  K_{2}^{(\ell_1,\ell_2)}\left(t_{k}-t_{m},u_{k}-u_{m}\right)$.
 For convenience, we write (\ref{eq:matrix_2d}) as 
\[
\begin{bmatrix}G^{(0,0)}  & G_{1}^{T}\\
G_{1} & G_{2}
\end{bmatrix}\begin{bmatrix}\mathbf{a}\\
\tilde{\mathbf{b}}
\end{bmatrix}=\begin{bmatrix}\mathbf{v}\\
0
\end{bmatrix},
\]
where
\[
G_{2}:=\begin{bmatrix}G^{(2,0)} & G^{(1,1)}\\
G^{(1,1)}  & G^{(0,2)}
\end{bmatrix},\quad G_{1}:=\begin{bmatrix}G^{(1,0)} \\
G^{(0,1)} 
\end{bmatrix},\quad\tilde{\mathbf{b}}:=\begin{bmatrix}\mathbf{b}\\
\mathbf{c}
\end{bmatrix}.
\]
 We begin by showing that the matrix $G_{2}$ is invertible for sufficiently
large $\nu$. $G_{2}$ is invertible if both $G^{(0,2)}$ and its
Schur complement $G_{2}^{s}:=G^{(2,0)}-G^{(1,1)}\left(G^{(0,2)}\right)^{-1}G^{(1,1)}$
are invertible. Using the properties of the bivariate admissible kernel (see Definition \ref{def:2d_kernel}), we
observe that 
\begin{eqnarray*} 
\left\Vert K_{2}^{\left(0,2\right)}(0,0)I-G^{(0,2)}\right\Vert _{\infty} & := & \max_{k}\sum_{m\neq k}\left|K_{2}^{\left(0,2\right)}\left(t_{k}-t_{m},u_{k}-u_{m}\right)\right|\\
 & \leq & C_{0,2}\sum_{m\neq k}\frac{1}{\left(1+\left(t_{k}-t_{m}\right)^{2}+\left(u_{k}-u_{m}\right)^{2}\right)^{3/2}}.\nonumber 
\end{eqnarray*}
According to (\ref{eq:ring}), the $n^{th}$ `rectangular ring' with respect to $(t_k,u_k)$ contains at most $9n$ elements of $T_2$. So under the separation condition
of Definition \ref{def:separation_2d}, we get 
\begin{equation}
\left\Vert K_{2}^{\left(0,2\right)}(0,0)I-G^{(0,2)}\right\Vert _{\infty}\leq C_{0,2}\sum_{n=1}^{\infty}\frac{9n}{\left(1+n^{2}\nu^{2}\right)^{3/2}}=C_{0,2}E_{2}\left(\nu\right),\label{eq:G2uu}
\end{equation}
where 
\begin{equation} \label{eq:E2}
E_{2}\left(\nu\right):=\sum_{n=1}^{\infty}\frac{9n}{\left(1+n^{2}\nu^{2}\right)^{3/2}}\leq \frac{3\pi^2}{2\nu^3}.
\end{equation}
Therefore, if $\nu$ is chosen such that
\begin{equation} \label{cond-nu-biv}
\nu^3 \ge \frac{3\pi^2C_{0,2}}{2\left|{K_2^{(0,2)}(0,0)}\right|},
\end{equation}
then, $\left\Vert K_{2}^{\left(0,2\right)}(0,0)I-G^{(0,2)}\right\Vert _{\infty}<\left|K_{2}^{\left(0,2\right)}(0,0)\right|$,
and $G^{(0,2)}$ is invertible. The Schur complement of $G^{(0,2)}$ can be bounded by
\begin{equation}
\left\Vert K_{2}^{\left(2,0\right)}(0,0)I-G_2^{s}\right\Vert _{\infty}\leq\left\Vert K_{2}^{\left(2,0\right)}(0,0)I-G^{(2,0)}\right\Vert _{\infty} +\left\Vert G^{(1,1)}\right\Vert _{\infty}^2\left\Vert \left(G^{(0,2)}\right)^{-1}\right\Vert _{\infty}.\label{eq:G2s}
\end{equation}
Using the same considerations as in (\ref{eq:G2uu}) we have 
\[
\left\Vert K_{2}^{\left(2,0\right)}(0,0)I-G^{(2,0)}\right\Vert _{\infty}\leq C_{2,0}E_{2}\left(\nu\right),
\]
and since $K_{2}^{\left(1,1\right)}(0,0)=0,$
we also have 
\begin{equation} \label{eq:G2_ut}
\left\Vert G^{(1,1)}\right\Vert _{\infty}\leq C_{1,1}E_{2}\left(\nu\right).
\end{equation}
 Substituting into (\ref{eq:G2s}) and using (\ref{15}), we get 
\begin{eqnarray}
\left\Vert K_{2}^{\left(2,0\right)}(0,0)I-G_{2}^{s}\right\Vert _{\infty} & \leq & C_{2,0}E_{2}\left(\nu\right)+\frac{C_{1,1}^{2}E_{2}^{2}\left(\nu\right)}{\left|K_{2}^{\left(0,2\right)}(0,0)\right|-C_{0,2}E_{2}\left(\nu\right)}\nonumber \\
 & \leq & C_{2,0}E_{2}\left(\nu\right)\left(1+\frac{C_{1,1}^{2}E_{2}\left(\nu\right)}{C_{2,0}\left(\left|K_{2}^{\left(0,2\right)}(0,0)\right|-C_{0,2}E_{2}\left(\nu\right)\right)}\right)\label{eq:G2s-1}\\
 & \leq & 2C_{2,0}E_{2}\left(\nu\right),\nonumber 
\end{eqnarray}
where the last inequality holds for,
\begin{equation} \label{cond-nu-biv-2}
\nu^3 \ge \frac{3\pi^2(C_{1,1}^2+C_{2,0}C_{0,2})}{2C_{2,0}\left| {K_2^{(0,2)}(0,0)} \right|}.
\end{equation}
Similarly to  (\ref{cond-nu-biv}), if we impose,
\begin{equation} \label{cond-nu-biv-3}
\nu^3 > \frac{3\pi^2C_{2,0}}{\left| {K_2^{(2,0)}(0,0)} \right|},
\end{equation}
then $2C_{2,0}E_{2}\left(\nu\right)< \left\vert  K_{2}^{\left(2,0\right)}(0,0)\right\vert $, and the invertibility of $G_{2}^{s}$ and  $G_{2}$  follows. 

In order to show that the matrix in (\ref{eq:matrix_2d}) is invertible, we need to show that the Schur complement of $G_2$ 
\begin{equation}
G_{s}:=G^{(0,0)}-G_{1}^{s}\left(G_{2}^{s}\right)^{-1}G_{1}^{s}-G^{(0,1)}\left(G^{(0,2)}\right)^{-1}G^{(0,1)},\label{eq:Gs}
\end{equation}
where 
\begin{equation}
G_{1}^{s}:=G^{(1,0)}-G^{(1,1)}\left(G^{(0,2)}\right)^{-1}G^{(0,1)},\label{eq:G1s}
\end{equation}
is invertible as well (see e.g. \cite{zhang2006schur}). We use the same considerations as before, and since $K_{2}^{\left(0,1\right)}\left(0,0\right)=K_{2}^{\left(1,0\right)}\left(0,0\right)=0$,
we get 
\begin{equation} \label{eq:Gu}
\left\Vert G^{(1,0)}\right\Vert_\infty \leq C_{1,0}E_{2}\left(\nu\right),
\quad\left\Vert G^{(0,1)}\right\Vert_\infty \leq C_{0,1}E_{2}\left(\nu\right).
\end{equation}
Substituting  (\ref{15}),  (\ref{eq:G2uu}), (\ref{eq:G2_ut}) and (\ref{eq:Gu}) into (\ref{eq:G1s}) leads to
\begin{eqnarray}
\left\Vert G_{1}^{s}\right\Vert _{\infty} & \leq & \left\Vert G^{(1,0)}\right\Vert _{\infty}+\left\Vert G^{(1,1)}\right\Vert _{\infty}\left\Vert \left(G^{(0,2)}\right)^{-1}\right\Vert _{\infty}\left\Vert G^{(0,1)}\right\Vert _{\infty}\nonumber \\
 & \leq & C_{1,0}E_{2}\left(\nu\right)+\frac{C_{1,1}C_{0,1}E_{2}^{2}\left(\nu\right)}{\left|K_{2}^{\left(0,2\right)}(0,0)\right|-C_{0,2}E_{2}\left(\nu\right)}\label{eq:G1s-1}\\
 & \leq & 2C_{1,0}E_{2}\left(\nu\right),\nonumber 
\end{eqnarray}
where the last inequality holds for $E_{2}\left(\nu\right)\leq\frac{C_{1,0}\left|K_{2}^{\left(0,2\right)}(0,0)\right|}{C_{1,1}C_{0,1}+C_{1,0}C_{0,2}}$.
Using the estimate 
\begin{equation*}
\left\Vert K_{2}\left(0,0\right)I-G^{(0,0)}\right\Vert _{\infty}\leq C_{0,0}E_{2}\left(\nu\right),\label{eq:G0}
\end{equation*}
and substituting (\ref{15}), (\ref{eq:G2uu}), (\ref{eq:G2s-1}), (\ref{eq:Gu}) and (\ref{eq:G1s-1}) into (\ref{eq:Gs})
 we obtain
\begin{eqnarray}
\left\Vert K_{2}\left(0,0\right)I-G_{s}\right\Vert _{\infty} & \leq & \left\Vert K_{2}\left(0,0\right)I-G^{(0,0)}\right\Vert _{\infty}+\left\Vert G_{1}^{s}\right\Vert _{\infty}^{2}\left\Vert \left(G_{2}^{s}\right)^{-1}\right\Vert _{\infty}+\left\Vert G^{(0,1)}\right\Vert _{\infty}^{2}\left\Vert \left(G^{(0,2)}\right)^{-1}\right\Vert _{\infty}\nonumber \\
 & \leq & C_{0,0}E_{2}\left(\nu\right)+\frac{4C_{1,0}^{2}E_{2}^{2}\left(\nu\right)}{\left|K_{2}^{\left(2,0\right)}(0,0)\right|-2C_{2,0}E_{2}\left(\nu\right)}\nonumber \\
 &\qquad+&\frac{C_{0,1}^{2}E_{2}^{2}\left(\nu\right)}{\left|K_{2}^{\left(0,2\right)}(0,0)\right|-C_{0,2}E_{2}\left(\nu\right)}\label{eq:Gs-1}\\
 & \leq & 3C_{0,0}E_{2}\left(\nu\right),\nonumber 
\end{eqnarray}
where the last inequality holds for $E_{2}\left(\nu\right)\leq\min\left\{ \frac{C_{0,0}\left|K_{2}^{\left(2,0\right)}(0,0)\right|}{2\left(2C_{1,0}^{2}+C_{2,0}C_{0,0}\right)},\frac{C_{0,0}\left|K_{2}^{\left(0,2\right)}(0,0)\right|}{C_{0,1}^{2}+C_{0,0}C_{0,2}}\right\} .$
Thus, for sufficiently large $\nu$, $\left\Vert K_{2}\left(0,0\right)I-G_{s}\right\Vert _{\infty}< K_{2}\left(0,0\right)$,
and hence $G_{s}$ is invertible. Combining this result with (\ref{eq:G2uu})
and (\ref{eq:G2s-1}), we conclude that (\ref{eq:matrix_2d}) has
a unique solution. 

Going back to (\ref{eq:matrix_2d}), we use the inversion formula
to get \cite{bendory2013exact} 
\[
\begin{bmatrix}\mathbf{a}\\
\mathbf{b}\\
\mathbf{c}
\end{bmatrix}=\begin{bmatrix}I\\
-\left(G_{2}^{s}\right)^{-1}G_{1}^{s}\\
\left(G^{(0,2)}\right)^{-1}\left(G^{(1,1)}\left(G_{2}^{s}\right)^{-1}G_{1}^{s}-G^{(0,1)}\right)
\end{bmatrix}G_{s}^{-1}\mathbf{v},
\]
 so that 
\[
\mathbf{\left\Vert a\right\Vert _{\infty}\le}\left\Vert G_{s}^{-1}\right\Vert _{\infty}\leq\frac{1}{K_{2}\left(0,0\right)-3C_{0,0}E_{2}\left(\nu\right)},
\]
and 
\begin{eqnarray*}
\mathbf{\left\Vert b\right\Vert _{\infty}} & \leq & \left\Vert \left(G_{2}^{s}\right)^{-1}\right\Vert _{\infty}\left\Vert G_{1}^{s}\right\Vert _{\infty}\left\Vert G_{s}^{-1}\right\Vert _{\infty}\\
 & \leq & \frac{2C_{1,0}E_{2}\left(\nu\right)}{\left(\left|K_{2}^{\left(2,0\right)}(0,0)\right|-2C_{2,0}E_{2}\left(\nu\right)\right)\left(K_{2}\left(0,0\right)-3C_{0,0}E_{2}\left(\nu\right)\right)},
\end{eqnarray*}
and 
\begin{eqnarray*}
\mathbf{\left\Vert c\right\Vert _{\infty}} & \leq & \left\Vert \left(G^{(0,2)}\right)^{-1}\right\Vert _{\infty}\left(\left\Vert G^{(1,1)}\right\Vert _{\infty}\left\Vert \left(G_{2}^{s}\right)^{-1}\right\Vert _{\infty}\left\Vert G_{1}^{s}\right\Vert _{\infty}+\left\Vert G^{(0,1)}\right\Vert _{\infty}\right)\left\Vert G_{s}^{-1}\right\Vert _{\infty}\\
 & \leq & \frac{\frac{2C_{1,1}C_{1,0}E_{2}^{2}\left(\nu\right)}{\left|K_{2}^{\left(2,0\right)}(0,0)\right|-2C_{2,0}E_{2}\left(\nu\right)}+C_{0,1}E_{2}\left(\nu\right)}{\left(\left|K_{2}^{\left(0,2\right)}(0,0)\right|-C_{0,2}E_{2}\left(\nu\right)\right)\left(K_{2}\left(0,0\right)-3C_{0,0}E_{2}\left(\nu\right)\right)}\\
 & \leq & \frac{2C_{0,1}E_{2}\left(\nu\right)}{\left(\left|K_{2}^{\left(0,2\right)}(0,0)\right|-C_{0,2}E_{2}\left(\nu\right)\right)\left(K_{2}\left(0,0\right)-3C_{0,0}E_{2}\left(\nu\right)\right)},
\end{eqnarray*}
where the last inequality holds for $E_{2}\left(\nu\right)\leq\frac{C_{0,1}\left|K_{2}^{\left(2,0\right)}(0,0)\right|}{2\left(C_{1,1}C_{1,0}+C_{0,1}C_{2,0}\right)}$. This proves (\ref{eq:a}),(\ref{eq:b}) and (\ref{eq:c}).

If $v_k=1$, similarly to (\ref{eq:ak1}) we conclude that 
\begin{eqnarray}
a_{k} & = &\left(G_{s}^{-1}\mathbf{v}\right)_{k}\nonumber \\
 & \text{\ensuremath{\geq}} & \frac{1}{K_{2}\left(0,0\right)}\left(1-\left\Vert G_{s}^{-1}\right\Vert _{\infty}\left\Vert K_{2}\left(0,0\right)I-G_{s}\right\Vert _{\infty}\right)\label{eq:ak}\\
 & \geq & \frac{1}{K_{2}\left(0,0\right)}\left(1-\frac{3C_{0,0}E_{2}\left(\nu\right)}{K_{2}\left(0,0\right)-3C_{0,0}E_{2}\left(\nu\right)}\right).\nonumber 
\end{eqnarray}

\subsubsection{Proof of Lemma \ref{lemma:2d_1}}

Fix $(t,u)\in S_k$ with respect to $(t_k,u_k)\in T_2$ (see (\ref{eq:S})), and assume that $q(t_k)=1$. The proof is similar for the case $q(t_k)=-1$. Since  $\left| {t-t_{m} } \right|>\nu /2$ or $\left| {u-u_{m} } \right|>\nu /2$ for $m\ne k$, we have, using the separation assumption, that for $\ell_1+\ell_2\leq 3$ (compare with (\ref{eq:E2})): 

\begin{align} \label{eq:sum_K}
 \sum\limits_{m\ne k} {\left| {K_2^{\left( \ell_1,\ell_2 \right)}\left( {t-t_{m},u-u_{m} } 
\right)} \right|} & \le C_{\ell_1,\ell_2} \sum\limits_{m\ne k}\frac{1}{\left(1+\left( {t-t_{m} 
} \right)^{2}+\left(u-u_m\right)^{2}\right)^\frac{3}{2}} \nonumber\\ 
& \le C_{\ell_1,\ell_2} \sum\limits_{n=1}^\infty {\frac{9n}{\left(1+\left( {2^{-1} n\nu } \right)^{2}\right)^\frac{3}{2}}} \\ 
& \le \frac{12C_{\ell_1,\ell_2}\pi^2}{\nu^3}. \nonumber
\end{align}
We start by proving that the Hessian of $q(t,u)$ is negative definite.
Recall that the Hessian of $q(t,u)$ is given by 
\[
H(q)(t,u):=\begin{bmatrix}q^{(2,0)}\left(t,u\right) & q^{(1,1)}\left(t,u\right)\\
q^{(1,1)}\left(t,u\right) & q^{(0,2)}\left(t,u\right)
\end{bmatrix}.
\]

 By (\ref{eq:q2def}) we have 
\begin{align*}
q^{\left(2,0\right)}\left(t,u\right) & = \sum_{m}\left(a_{m}K_{2}^{\left(2,0\right)}\left(t-t_{m},u-u_{m}\right) \right.\\
&\qquad\left.+b_{m}K_{2}^{\left(3,0\right)}\left(t-t_{m},u-u_{m}\right)+c_{m}K_{2}^{\left(2,1\right)}\left(t-t_{m},u-u_{m}\right)\right)\\
 & \leq  a_{k}K_{2}^{\left(2,0\right)}\left(t-t_k,u-u_k\right)
 +\Vert \mathbf{a}\Vert_\infty  \sum_{m\neq k}\left\vert K_{2}^{\left(2,0\right)}\left(t-t_{m},u-u_{m}\right)\right\vert  \\ 
&\qquad  +\Vert \mathbf{b}\Vert_\infty  \sum_{m} \left\vert K_{2}^{\left(3,0\right)}\left(t-t_{m},u-u_{m}\right)\right\vert+\Vert \mathbf{c}\Vert_\infty  \sum_{m} \left\vert K_{2}^{\left(2,1\right)}\left(t-t_{m},u-u_{m}\right)\right\vert.
\end{align*}
Using the local convexity of the bivariate kernel, (\ref{eq:a}), (\ref{eq:ak2}) and (\ref{eq:sum_K}) we get
\begin{equation*}
\begin{split}
q^{\left(2,0\right)}\left(t,u\right) &\leq \frac{-\beta}{K_{2}\left(0,0\right)}\left(1-\frac{9\pi^2C_{0,0}}{2\nu^3K_{2}\left(0,0\right)-9\pi^2C_{0,0}}\right)+\frac{24\pi^2C_{2,0}}{ K_{2}\left(0,0\right)\nu^3-9\pi^2C_{0,0}}\\
&+\left(\Vert\mathbf{b}\Vert_\infty C_{3,0}+\Vert\mathbf{c}\Vert_\infty C_{2,1}\right)\left(1+\frac{12\pi^2}{\nu^3}\right).
\end{split}
\end{equation*}
Hence, using (\ref{eq:b}) and (\ref{eq:c}) it is evident that for sufficiently large $\nu$ we have $q^{\left(2,0\right)}\left(t,u\right)\leq -\frac{\beta}{2K_{2}\left(0,0\right)} $. Plainly, similar argument holds for $q^{\left(0,2\right)}\left(t,u\right)$ as well.

Next, we consider $\vert q^{\left(1,1\right)}\left(t,u\right)\vert $. By (\ref{eq:q2def}) we have
\begin{align*}
\left|q^{\left(1,1\right)}\left(t,u\right)\right| &\leq a_k \left| K_{2}^{\left(1,1\right)}\left(t-t_{k},u-u_{k}\right)\right|
+\Vert\mathbf{a}\Vert _\infty \sum_{m\neq k}\left|K_{2}^{\left(1,1\right)}\left(t-t_{m},u-u_{m}\right)\right|\\
 &+\Vert\mathbf{b}\Vert_\infty\sum_m \left| K_{2}^{\left(2,1\right)}\left(t-t_{m},u-u_{m}\right)\right| \\
 &+\Vert\mathbf{c}\Vert_\infty\sum_m\left| K_{2}^{\left(1,2\right)}\left(t-t_{m},u-u_{m}\right)\right|.
\end{align*}

Observe that $K_{2}^{(1,1)}\left(0,0\right)=0$, so 
\begin{align*}
\left| K_{2}^{\left(1,1\right)}\left(t-t_{k},u-u_{k}\right)\right|&=\frac{\left| K_{2}^{\left(1,1\right)}\left(t-t_{k},u-u_{k}\right)-K_{2}^{\left(1,1\right)}\left(0,0\right)\right|}{\max\{\vert t-t_k\vert,\vert u-u_k\vert  \}}\max\{\vert t-t_k\vert,\vert u-u_k\vert  \} \\
&\leq \max\{C_{1,2},C_{2,1}\}\varepsilon_1.
\end{align*}

Consequently, we obtain 

\begin{align*}\label{eq:qut}
\left|q^{\left(1,1\right)}\left(t,u\right)\right| &\leq \frac{2\nu^3\max\{C_{1,2},C_{2,1}\}}{2K_{2}\left(0,0\right)\nu^3-9\pi^2C_{0,0}}\varepsilon_1
+\frac{24C_{1,1}\pi^2}{2K_{2}\left(0,0\right)\nu^3-9\pi^2C_{0,0}}\\
 &+\left(\Vert\mathbf{b}\Vert_\infty C_{2,1}+\Vert\mathbf{c}\Vert_\infty C_{1,2}\right)\left(1+\frac{12\pi^2}{\nu^3}\right).
 \end{align*}
Hence, for sufficiently large $\nu$ and sufficiently small $\varepsilon_1$, $\left|q^{\left(1,1\right)}\left(t,u\right)\right| <\frac{\beta}{2K_2(0,0)}$. Consequently, the determinant of the Hessian is positive 
\begin{align*}
&q^{\left(2,0\right)}\left(t,u\right)q^{\left(0,2\right)}\left(t,u\right)-\left(q^{\left(1,1\right)}\right)^2\left(t,u\right)>0,
\end{align*}
whereas the trace is negative
\begin{align*}
q^{\left(2,0\right)}\left(t,u\right)+q^{\left(0,2\right)}\left(t,u\right)<0.
\end{align*}
 As a result, both eigenvalues of the Hessian are negative, so the Hessian is negative definite for any $(t,u)\in S$. Using the Taylor remainder theorem (similarly to (\ref{eq:taylor})) we conclude that  $q(t,u)<1$ for all $(t,u)\in S$.

To complete the proof, we need to show that $q(t,u)>-1$. Recall that $K_2(t,u)$ decreases as function of both variables in $0\leq t,u<\varepsilon_1$ (see Remark \ref{Remark3}). So,
\begin{equation}
\begin{split}
q(t,u)&\geq a_k K(t-t_k,u-u_k)- \Vert \mathbf{a}\Vert_\infty\sum_{m\neq k}\left\vert K(t-t_m,u-u_m)\right\vert \\
&- \Vert \mathbf{b}\Vert_\infty\sum_{m}\left\vert K^{(1,0)}(t-t_m,u-u_m)\right\vert -\Vert \mathbf{c}\Vert_\infty\sum_{m}\left\vert K^{(0,1)}(t-t_m,u-u_m)\right\vert \\
& \geq \frac{ K(\varepsilon_1,\varepsilon_1)}{K_{2}\left(0,0\right)}\left(1-\frac{9\pi^2C_{0,0}}{2\nu^3K_{2}\left(0,0\right)-9\pi^2C_{0,0}}\right)-\frac{24\pi^2C_{0,0}}{2K_{2}\left(0,0\right)\nu^3-9\pi^2C_{0,0}}\\
&-\left(C_{1,0}\Vert \mathbf{b}\Vert_\infty+C_{0,1}\Vert \mathbf{c}\Vert_\infty\right)\left(1+\frac{12\pi^2}{\nu^3}\right).
\end{split}
\end{equation}
Thus it is clear that for sufficiently large $\nu$, $q(t,u)>-1$. This completes the proof.

\subsubsection{Proof of Lemma \ref{lemma:2d_2}}
Fix $(t,u)\in S^C$ (see (\ref{eq:S})), and denote $(t_k,u_k):=\min \{ \max\{\vert t-t_m\vert,\vert u-u_m\vert\} \thinspace: \thinspace(t_m,u_m)\in T_2\}$. Then, from (\ref{eq:q2def}), (\ref{eq:sum_K}) and the properties of bivariate admissible kernel
\begin{align*}
\left\vert q\left(  t,u\right)  \right\vert  & \leq\left\Vert \mathbf{a}%
\right\Vert _{\infty}\left(  \left\vert K_2\left(  t-t_{k}\right)  \right\vert
+\sum_{m\neq k}%
\left\vert K_2\left(  t-t_{m},u-u_{m}\right)  \right\vert \right)  \\
&\qquad+\left\Vert \mathbf{b}\right\Vert _{\infty}\sum_{m}\left\vert K_2^{^{(1,0)}%
}\left(  t-t_{m},u-u_{m}\right)  \right\vert+\left\Vert \mathbf{c}\right\Vert _{\infty}\sum_{m}\left\vert K_2^{^{(0,1)}%
}\left(  t-t_{m},u-u_{m}\right)  \right\vert  \\
& \leq \frac{2 \max\left\{ K_2\left(
\varepsilon_1,0\right), K_2\left(
0,\varepsilon_1\right)\right\}\nu^3}{2K_{2}\left(0,0\right)\nu^3-9\pi^2C_{0,0}}  + \frac{24\pi^2C_{0,0}}{2K_{2}\left(0,0\right)\nu^3-9\pi^2C_{0,0}}\\
&\qquad +\left( C_{1,0}\left\Vert \mathbf{b}\right\Vert _{\infty}+C_{0,1}\left\Vert \mathbf{c}\right\Vert _{\infty}\right)\left(  \frac{1}{\left(1+\varepsilon_1^2\right)^{\frac{3}{2}}}+\frac{12\pi^2}{\nu^3}\right) .
 \end{align*}

By (\ref{eq:b}) and (\ref{eq:c}) and since $ \max\left\{ K_2\left(
\varepsilon_1,0\right), K_2\left(
0,\varepsilon_1\right)\right\}  <K_2\left(  0,0\right)
$ we conclude that for sufficiently large $\nu$,
$\left\vert q\left(  t,u\right)  \right\vert <1$ for all $(t,u)\in S^C$. This completes the proof.

\subsection{Proof of Theorem \ref{Theorem2}}

Let $\hat{x}$ be the solution of the optimization
problem (\ref{6}) with $\left\Vert \hat{x}\right\Vert _{1}%
\leq\left\Vert x\right\Vert _{1}$ and let $h\left[  k\right] : =\hat{x}\left[  k\right]  -x\left[  k\right]$.
We decompose $h$  as
\begin{equation*}
h =h_T+h_{T^c}, 
\end{equation*} 
where $h_T$ is the part of the sequence $h $ with support in $T:=\{k_m\}$. If $h_T=0$, then $h=0$. Otherwise, $h_{T^C}\neq 0$ which implies the contradiction $\left\Vert \hat{x}\right\Vert _{1}%
>\left\Vert x\right\Vert _{1}$. 

The discrete support of the delays is identified as $\{t_m\}=\{k_m/N\}$ and it satisfies the condition $|t_j-t_k|\ge \nu \sigma$,
for $j \neq k$. Therefore, the set $T_\sigma:=\{t_m/\sigma \}=\{k_m/N\sigma \}$, satisfies a separation condition with $\nu$.
We have shown that under this separation condition, there exists $q$ of the
form (\ref{11}), corresponding to the interpolating conditions $q(t_m/\sigma)=q(k_m/N\sigma)=sgn(h_{T}[k_m])$
(see (\ref{8})) and also satisfying $|q(t)|< 1$ for $t \notin T_\sigma$ (see (\ref{9})). Therefore, we have that

\[
q_\sigma(t):=q(t/\sigma) = \sum_m {a_m K_\sigma \left({t-\frac{k_m}{N}}\right) + b_m (K^{(1)})_\sigma \left( {t - \frac{k_m}{N}}\right)},
\]
satisfies the interpolation conditions
\begin{equation} \label{dis-q-int}
q_\sigma[k_m]:=q_\sigma\left({\frac{k_m}{N}}\right)=sgn(h_{T}[k_m]), \qquad \forall k_m \in T,
\end{equation}
and
\begin{equation} \label{dis-q-small}
|q_\sigma[k]|<1, \qquad \forall k \notin T.
\end{equation}
Then, with $K_{\sigma}^{1}[k]:= (K^{(1)})_\sigma(k/N)$
\begin{align}
\left\vert \sum_{k\in\mathbb{Z}}q_\sigma\left[  k\right]  h\left[  k\right]  \right\vert  &
=\left\vert \sum_{k\in\mathbb{Z}}\left(  \sum_{k_m\in T}\left(  a_{m}K_{\sigma}\left[
k-k_{m}\right]  +b_{m}K_{\sigma}^{1}\left[  k-k_{m}\right]
\right)  \right)  h\left[  k\right]  \right\vert \nonumber\\
&  \leq\left\Vert \mathbf{a}\right\Vert _{\infty}\sum_{k_m\in T}\left\vert \sum
_{k\in\mathbb{Z}}K_{\sigma}\left[  k-k_{m}\right]  h\left[  k\right]  \right\vert\nonumber \\
&+\left\Vert \mathbf{b}\right\Vert _{\infty}\sum_{k_m\in T}\left\vert \sum
_{k\in\mathbb{Z}}K_{\sigma}^{1}\left[  k-k_{m}\right]  h\left[  k\right]
\right\vert. \label{30}
\end{align}
Observe that (\ref{5}) and (\ref{6}) give
\begin{align}
\left\Vert K_{\sigma}\ast h
\right\Vert _{1} &  =\left\Vert K_{\sigma}\ast(\hat{x}-x)\right\Vert_1  \nonumber \\
& \leq \left\Vert y- K_{\sigma}\ast x\right\Vert_1 + \left\Vert y- K_{\sigma}\ast\hat{x}\right\Vert_1 \nonumber \\
&\leq 2\delta. \label{28}%
\end{align}
Now, using (\ref{28}) we have%
\[
\sum_{k_m\in T}\left\vert \sum_{k\in\mathbb{Z}}K_{\sigma}\left[  k-k_{m}\right]  h\left[
k\right]  \right\vert \leq\sum_{r\in\mathbb{Z}}\left\vert \sum_{k\in\mathbb{Z}}K_{\sigma}\left[
k-r\right]  h\left[  k\right]  \right\vert \leq2\delta.
\]
From the admissible kernel properties (see Definition \ref{Definition2}) we get%
\[
\left\vert K_{\sigma}^{1}\left[  k-k_{m}\right]  \right\vert
=\left\vert K^{\left(  1\right)  }\left(  \frac{k-k_{m}}{N\sigma}\right)
\right\vert \leq\frac{C_{1}}{1+\left(  \frac{k-k_{m}}{N\sigma}\right)  ^{2}}.%
\]
Using the separation condition $\left\vert k_{i}-k_{j} \right\vert \geq\nu N \sigma$, $\forall k_i,k_j \in T$, 
we can estimate for any $k$
\begin{equation*}
\sum_{k_m\in T}\frac{1}{1+\left(
\frac{k-k_{m}}{N\sigma}\right)  ^{2}}< 2\left(1 + E(\nu)\right),
\end{equation*} 
where $E(\nu):=\pi^2 / 6\nu^2$ (see (\ref{eq:E_def})).
Then, 
\begin{align*}
\sum_{k_m\in T}\left\vert \sum_{k\in\mathbb{Z}}K_{\sigma}^{1}\left[
k-k_{m}\right]  h\left[  k\right]  \right\vert  &  \leq C_{1}\sum
_{k\in\mathbb{Z}}\left\vert h\left[  k\right]  \right\vert \sum_{k_m\in T}\frac{1}{1+\left(
\frac{k-k_{m}}{N\sigma}\right)  ^{2}}\\
&  <2C_{1}(1+E\left(  \nu\right))\left\Vert h\right\Vert _{1}.
\end{align*}
Substituting in (\ref{30}) we get%
\[
\left\vert \sum_{k\in\mathbb{Z}}q_\sigma\left[  k\right]  h\left[  k\right]  \right\vert
\leq2\delta\left\Vert \mathbf{a}\right\Vert _{\infty}+2C_{1}(1+E\left(
\nu\right))\left\Vert\mathbf{b}\right\Vert _{\infty}\left\Vert h\right\Vert _{1}.
\]
On the other hand, from (\ref{dis-q-int}) and (\ref{dis-q-small}) we get
\begin{align*}
\left\vert \sum_{k\in\mathbb{Z}}q_\sigma\left[  k\right]  h\left[  k\right]  \right\vert  &
=\left\vert \sum_{k\in\mathbb{Z}}q_\sigma\left[  k\right]  \left(  h_{T}\left[  k\right]
+h_{T^{C}}\left[  k\right]  \right)  \right\vert \\
&  \geq\left\Vert h_{T}\right\Vert _{1}-\underset{k\in\mathbb{Z}\backslash
T}{\max}\left\vert q_\sigma\left[  k\right]  \right\vert \left\Vert h_{T^{C}%
}\right\Vert _{1}.
\end{align*}
Combining the two inequalities, we get%
\begin{equation}
\left\Vert h_{T}\right\Vert _{1}-\underset{k\in\mathbb{Z}\backslash T}%
{\max\left\vert q_\sigma\left[  k\right]  \right\vert }\left\Vert h_{T^{C}%
}\right\Vert _{1}\leq2\delta\left\Vert \mathbf{a}\right\Vert _{\infty
}+2C_{1}(1+E(\nu))\left\Vert \mathbf{b}\right\Vert _{\infty}\left\Vert h\right\Vert
_{1}.  \label{31}%
\end{equation}
Assume $\left\vert k-k_{m}\right\vert \leq\varepsilon N\sigma$, $k \ne k_m$, for some $k_m \in T$. By (\ref{eq:q_up}) we observe that
\begin{equation*}
\left\vert {q_\sigma[k]}\right\vert=\left\vert { q\left( {\frac{k}{N\sigma}} \right)} \right\vert \le  1 - 
\frac{\beta }{4K(0)\left(N\sigma\right)^2}.
\end{equation*}
For the case $\left\vert k-k_{m}\right\vert > \varepsilon N\sigma$, for all $k_m \in T$, we apply (\ref{eq:q_far}), to derive
\begin{equation*}
\left\vert {q_\sigma[k]}\right\vert=\left\vert { q\left( {\frac{k}{N\sigma}} \right)} \right\vert \le  1 - 
\frac{\beta \varepsilon^2}{4K(0)}.
\end{equation*}
Combining these last two estimates gives a uniform estimate for sufficiently large $\nu$
\begin{equation*}
\underset{k\in\mathbb{Z}\backslash
T}{\max}\left\vert q_\sigma\left[  k\right]  \right\vert \leq  1 - \frac{\beta }{4K(0)\gamma^2},
\end{equation*}
where $\gamma:=\max \{ N\sigma, \varepsilon^{-1}\}$. Substituting into (\ref{31}) we get%
\begin{equation} \label{two-side-2}
\left\Vert h_{T}\right\Vert _{1}-\left(1 - \frac{\beta}{4K(0)\gamma^2}\right)  \left\Vert h_{T^{C}}\right\Vert _{1}\leq
2\delta\left\Vert \mathbf{a}\right\Vert _{\infty}+2C_1(1+E\left(  \nu\right))\left\Vert \mathbf{b}%
\right\Vert _{\infty}\left\Vert h\right\Vert _{1}.
\end{equation}
We also have from (\ref{6})
\begin{align*}
\left\Vert x\right\Vert _{1} &  \geq\left\Vert x+h\right\Vert _{1}=\left\Vert
x+h_{T}\right\Vert _{1}+\left\Vert h_{T^{C}}\right\Vert _{1}\\
&  \geq\left\Vert x\right\Vert _{1}-\left\Vert h_{T}\right\Vert _{1}%
+\left\Vert h_{T^{C}}\right\Vert _{1},%
\end{align*}
leading to
\[
\left\Vert h_{T^{C}}\right\Vert _{1}\leq\left\Vert h_{T}\right\Vert _{1}.%
\]
Applying this with (\ref{two-side-2}) yields 
\begin{align*}
\left\Vert h\right\Vert _{1} &  =\left\Vert h_{T^{C}}\right\Vert
_{1}+\left\Vert h_{T}\right\Vert _{1} \\
&\leq2\left\Vert h_{T}\right\Vert _{1}\\
&  \leq\frac{8K(0)\gamma^2}{\beta}\left(  \delta\left\Vert
\mathbf{a}\right\Vert _{\infty}+C_{1}(1+E\left(  \nu\right))\left\Vert \mathbf{b}\right\Vert _{\infty
}\left\Vert h\right\Vert _{1} \right).
\end{align*}
This gives
\[
\left\| h \right\|_{1} \le \frac{8K\left( 0 \right)\gamma^2 \left\| a \right\|_{\infty } }{\beta -8K\left( 0 \right)
\gamma^2 C_{1} \left( {1+E\left( \nu \right)} \right)\left\| b 
\right\|_{\infty } }\delta \]
Then using (\ref{23}),(\ref{24}),(\ref{13}) and (\ref{eq:E_def}) we get 
\begin{equation*}
\left\Vert h\right\Vert _{1}\leq\frac{72K(0)\left\vert K^{(2)}(0)\right\vert\gamma^2}{9\beta K(0)\left\vert K^{(2)}(0)\right\vert-D_1\nu^{-2}-D_2\nu^{-4}}\delta,
\end{equation*}
where 
\begin{eqnarray*}
D_{1} & = & 3\pi^2 \left(C_2\beta +2C_0\beta  +8C_1^2K(0)\gamma^2 \right) ,\\
D_{2} & = &4\pi^4C_1^2K(0)\gamma^2.
\end{eqnarray*}
This completes the proof.

\section{Numerical Experiments} \label{Section4}
We performed extensive numerical experiments to validate the theoretical results.
All experiments approximate the TV minimization by solving the appropriate $\ell_{1}$ minimization using CVX  \cite{cvx}.
The signals were generated in two steps. First, random locations were sequentially added to the signal's
support in the interval $[-1,1]$ with discretization step of 0.01, while keeping the separation condition. 
 Once the support was determined, the amplitudes were drawn randomly
from an i.i.d normal distribution with standard deviation suits to the desired signal to noise (SNR) ratio.

The first experiment aims to estimate empirically the minimal separation constant $\nu$  (see Definition \ref{Definition1}) in a noise-free environment for different admissible kernels. As can be seen in Figure \ref{fig1}, for Cauchy kernel it suffices to set $\nu=0.45$, whereas Gaussian kernel requires separation constant of $\nu=1.1$ (see also Table \ref{table1}).
\begin{figure}[h!]
\begin{centering}
\includegraphics[scale=0.5]{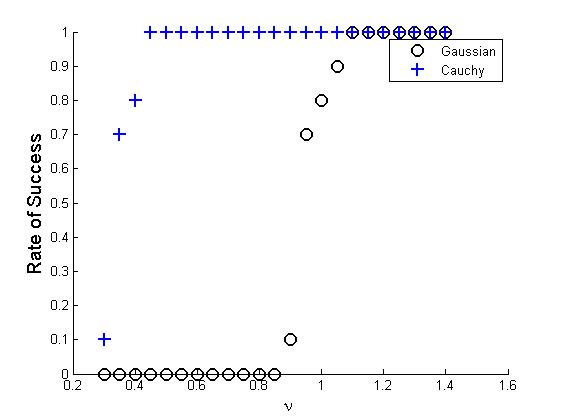}
\par\end{centering}

\protect\caption{Support detection as a function of the separation constant $\nu$ for Cauchy and Gaussian kernels. For each
value of $\nu$, 10 experiments were conducted with $\sigma=0.1$. An experiment is called
a success if it detects the signal's support precisely. Rate of success
1 means that the support was located precisely for all 10 experiments. }
 \label{fig1}
\end{figure}

Figure \ref{fig:example} presents two examples for atomic decomposition of stream of Cauchy delays with $\sigma=0.1$ and  $\nu=0.7$, contaminated with Gaussian noise (SNRs of $17.9$ and $27.5$ db). We mention that in order to achieve good recovery results we had to increase the separation constant. As can be seen, the solution in some cases misses the small delays of the signal, which are at the level of the noise, but manages to recover the larger delays with high accuracy. 
 

\begin{figure}[h!]
\begin{center}$
\begin{array}{cc}
\includegraphics[scale=0.45]{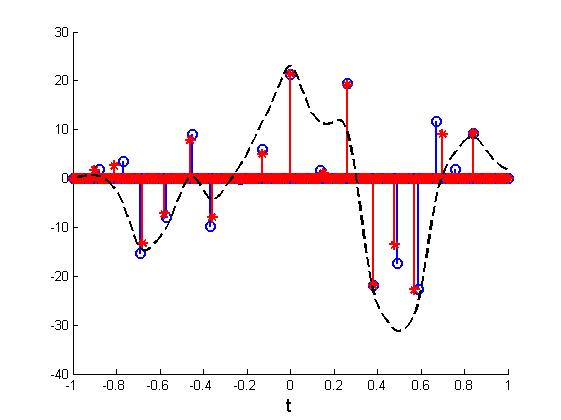}&
\includegraphics[scale=0.45]{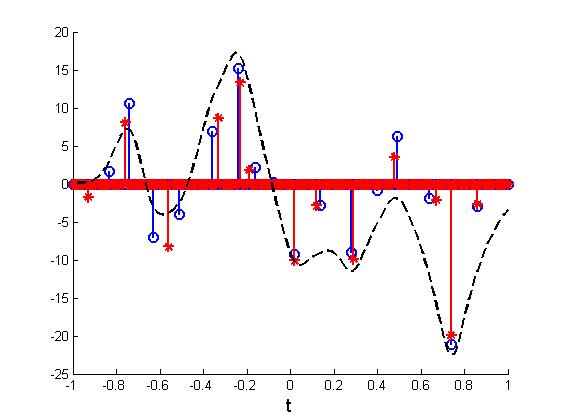} \\
\mbox{(a)} & \mbox{(b)}
\end{array}$
\end{center}
\protect\caption{\label{fig:example}Atomic decomposition of stream of Cauchy delays
with $\sigma=0.1$ and a separation constant of $\nu=0.7$ from noisy measurements with: (a) $\delta=20$ (SNR$=27.5$
db) (b) $\delta=60$ (SNR$=17.9$
db).  The black line represents the noisy measurements, 
the blue circles represent the original signal, and the red asterisks
the recovered signal.}
\end{figure}


\section{Conclusions}\label{Section5}

In this paper, we have shown that  a standard convex optimization technique can robustly decompose a stream of pulses into its atoms. The localization properties of the decomposition were derived in \cite{bendory2015stable}. This holds provided that the distance between the atoms is of the order of $\sigma$ (see (\ref{1})), and that the kernel satisfies mild localization conditions.
 In contrast to previous works, our method is stable and relies on theoretical results on the continuum, implying that there is no limitation on the discretization step. In a consecutive paper \cite{bendory2015robust}, it was proven that the separation is unnecessary if the underlying signal is known to be positive (i.e. $c_m>0$). In this case, the separation can be replaced by a weaker condition of Rayleigh regularity.

Our results show that the minimal separation needed for the success of the recovery depends on the global and the local properties of the kernel. It seems that the degree of the kernel's concavity near the origin has a particular importance, i.e. a `flat' kernel near the origin requires higher separation.

We have showed explicitly that our technique applies to univariate and bivariate signals. We strongly believe that this result holds in higher dimensions since parts of the proof can be easily generalized to any dimension. However, there
are certain technical challenges which we hope to overcome in future work.

This work is part of an ongoing effort to prove and demonstrate the effectiveness of convex optimization techniques to  robustly recover signals from their projections onto polynomial spaces \cite{candes2013towards,bendory2015super,bendory2013Legendre} and from their convolution with known kernels. The projection of signals onto spaces generated by shifts of one function or shifts and dilations of one function were investigated extensively in the  literature \cite{villemoes1994wavelet,deboor1994structure,de1994approximation,jia1993stability,filippov1995representation,cavaretta1991stationary,novikov2011wavelet,buhmann1999identifying,pinkus2013smoothness,terekhin2009affine} and found many applications (see for instance \cite{mallat1999wavelet,eldar2015sampling}). An interesting question is whether similar convex optimization techniques can be applied for the recovery of signals from these projections. We leave this question for a future research.

\subsection*{Acknowledgements}
The authors thank Yonina Eldar and Gongguo Tang for helpful discussions and to the referees for valuable remarks.

\bibliographystyle{plain}
\bibliography{bib}

\begin{thebibliography}{10}

\bibitem{azais2015spike}
J.~Azais, Y.~De~Castro, and F.~Gamboa.
\newblock Spike detection from inaccurate samplings.
\newblock {\em Applied and Computational Harmonic Analysis}, 38(2):177--195,
  2015.

\bibitem{Bar-IlanSub-Nyquis}
O.~Bar-Ilan and Y.C. Eldar.
\newblock Sub-nyquist radar via doppler focusing.
\newblock {\em IEEE Transactions on Signal Processing}, 62(7):1796--1811, 2014.

\bibitem{beck2014introduction}
A.~Beck.
\newblock {\em Introduction to Nonlinear Optimization: Theory, Algorithms, and
  Applications with MATLAB}, volume~19.
\newblock SIAM, 2014.

\bibitem{bendory2015robust}
T.~Bendory.
\newblock Robust recovery of positive stream of pulses.
\newblock {\em arXiv preprint arXiv:1503.08782}, 2015.

\bibitem{bendory2015stable}
T.~Bendory, A.~Bar-Zion, D.~Adam, S.~Dekel, and A.~Feuer.
\newblock Stable support recovery of stream of pulses with application to
  ultrasound imaging.
\newblock {\em To appear in IEEE Transactions on Signal Processing}, 2016.

\bibitem{bendory2013Legendre}
T.~Bendory, S.~Dekel, and A.~Feuer.
\newblock Exact recovery of non-uniform splines from the projection onto spaces
  of algebraic polynomials.
\newblock {\em Journal of Approximation Theory}, 182(0):7 -- 17, 2014.

\bibitem{bendory2013exact}
T.~Bendory, S.~Dekel, and A.~Feuer.
\newblock Exact recovery of dirac ensembles from the projection onto spaces of
  spherical harmonics.
\newblock {\em Constructive Approximation}, 42:183--207, 2015.

\bibitem{bendory2015super}
T.~Bendory, S.~Dekel, and A.~Feuer.
\newblock Super-resolution on the sphere using convex optimization.
\newblock {\em Signal Processing, IEEE Transactions on}, 63(9):2253--2262,
  2015.

\bibitem{bendory2015recovery}
T.~Bendory and Y.C. Eldar.
\newblock Recovery of sparse positive signals on the sphere from low resolution
  measurements.
\newblock {\em Signal Processing Letters, IEEE}, 22(12):2383--2386, 2015.

\bibitem{bhaskar2011atomic}
B.N. Bhaskar, G.~Tang, and B.~Recht.
\newblock Atomic norm denoising with applications to line spectral estimation.
\newblock {\em Signal Processing, IEEE Transactions on}, 61(23):5987--5999,
  2013.

\bibitem{buhmann1999identifying}
M.~Buhmann and A.~Pinkus.
\newblock Identifying linear combinations of ridge functions.
\newblock {\em Advances in Applied Mathematics}, 22(1):103--118, 1999.

\bibitem{candes2013super}
E.J. Cand{\`e}s and C.~Fernandez-Granda.
\newblock Super-resolution from noisy data.
\newblock {\em Journal of Fourier Analysis and Applications}, 19(6):1229--1254,
  2013.

\bibitem{candes2013towards}
E.J. Cand{\`e}s and C.~Fernandez-Granda.
\newblock Towards a mathematical theory of super-resolution.
\newblock {\em Communications on Pure and Applied Mathematics}, 2013.

\bibitem{cavaretta1991stationary}
A.~Cavaretta, W.~Dahmen, and C.~Micchelli.
\newblock {\em Stationary subdivision}, volume 453.
\newblock American Mathematical Soc., 1991.

\bibitem{chi2011sensitivity}
Y.~Chi, L.L. Scharf, A.~Pezeshki, and A.R. Calderbank.
\newblock Sensitivity to basis mismatch in compressed sensing.
\newblock {\em Signal Processing, IEEE Transactions on}, 59(5):2182--2195,
  2011.

\bibitem{de1994approximation}
C.~De~Boor, R.A. DeVore, and A.~Ron.
\newblock Approximation from shift-invariant subspaces of ${L}_2({R}^d)$.
\newblock {\em Transactions of the American Mathematical Society},
  341(2):787--806, 1994.

\bibitem{deboor1994structure}
C.~De~Boor, R.A. DeVore, and A.~Ron.
\newblock The structure of finitely generated shift-invariant spaces in
  ${L}_2({R}^d)$.
\newblock {\em Journal of Functional Analysis}, 119(1):37--78, 1994.

\bibitem{de2012exact}
Y.~De~Castro and F.~Gamboa.
\newblock Exact reconstruction using beurling minimal extrapolation.
\newblock {\em Journal of Mathematical Analysis and applications},
  395(1):336--354, 2012.

\bibitem{de2015exact}
Y.~De~Castro, F.~Gamboa, D.~Henrion, and J.B. Lasserre.
\newblock Exact solutions to super resolution on semi-algebraic domains in
  higher dimensions.
\newblock {\em arXiv preprint arXiv:1502.02436}, 2015.

\bibitem{de2014non}
Y.~De~Castro and G.~Mijoule.
\newblock Non-uniform spline recovery from small degree polynomial
  approximation.
\newblock {\em Journal of Mathematical Analysis and Applications}, 2015.

\bibitem{donoho2006compressed}
D.L. Donoho.
\newblock Compressed sensing.
\newblock {\em Information Theory, IEEE Transactions on}, 52(4):1289--1306,
  2006.

\bibitem{dragotti2007sampling}
P.L. Dragotti, M.~Vetterli, and T.~Blu.
\newblock Sampling moments and reconstructing signals of finite rate of
  innovation: Shannon meets strang--fix.
\newblock {\em Signal Processing, IEEE Transactions on}, 55(5):1741--1757,
  2007.

\bibitem{dumitrescu2007positive}
B.~Dumitrescu.
\newblock {\em Positive trigonometric polynomials and signal processing
  applications}.
\newblock Springer, 2007.

\bibitem{duval2015exact}
V.~Duval and G.~Peyr{\'e}.
\newblock Exact support recovery for sparse spikes deconvolution.
\newblock {\em Foundations of Computational Mathematics}, pages 1--41, 2015.

\bibitem{duval2015sparse}
V.~Duval and G.~Peyr{\'e}.
\newblock Sparse spikes deconvolution on thin grids.
\newblock {\em arXiv preprint arXiv:1503.08577}, 2015.

\bibitem{elad2010sparse}
M.~Elad.
\newblock {\em Sparse and redundant representations: from theory to
  applications in signal and image processing}.
\newblock Springer, 2010.

\bibitem{eldar2015sampling}
Y.C. Eldar.
\newblock {\em Sampling Theory: Beyond Bandlimited Systems}.
\newblock Cambridge University Press, 2015.

\bibitem{fernandez2015super}
Carlos Fernandez-Granda.
\newblock Super-resolution of point sources via convex programming.
\newblock {\em arXiv preprint arXiv:1507.07034}, 2015.

\bibitem{filippov1995representation}
V.I. Filippov and P.~Oswald.
\newblock Representation in lp by series of translates and dilates of one
  function.
\newblock {\em Journal of Approximation Theory}, 82(1):15--29, 1995.

\bibitem{cvx}
M.~Grant and S.~Boyd.
\newblock {CVX}: Matlab software for disciplined convex programming, version
  2.1, March 2014.

\bibitem{hornmatrix}
R.A. Horn and C.R. Johnson.
\newblock Matrix analysis. 1985.
\newblock {\em Cambridge}.

\bibitem{hua1990matrix}
Y.~Hua and T.K. Sarkar.
\newblock Matrix pencil method for estimating parameters of exponentially
  damped/undamped sinusoids in noise.
\newblock {\em Acoustics, Speech and Signal Processing, IEEE Transactions on},
  38(5):814--824, 1990.

\bibitem{jia1993stability}
R.Q. Jia and J.~Wang.
\newblock Stability and linear independence associated with wavelet
  decompositions.
\newblock {\em Proceedings of the American mathematical society},
  117(4):1115--1124, 1993.

\bibitem{liao2014music}
W.~Liao and A.~Fannjiang.
\newblock Music for single-snapshot spectral estimation: Stability and
  super-resolution.
\newblock {\em Applied and Computational Harmonic Analysis}, 2014.

\bibitem{mallat1999wavelet}
S.~Mallat.
\newblock {\em A wavelet tour of signal processing}.
\newblock Academic press, 1999.

\bibitem{mishra2015}
K.V. Mishra, C.~Myung, A.~Kruger, and X.~Weiyu.
\newblock Spectral super-resolution with prior knowledge.
\newblock {\em Signal Processing, IEEE Transactions on}, 63(20):5342--5357,
  2015.

\bibitem{moitra2014threshold}
A.~Moitra.
\newblock The threshold for super-resolution via extremal functions.
\newblock {\em arXiv preprint arXiv:1408.1681}, 2014.

\bibitem{novikov2011wavelet}
I.~Novikov, V.~Protasov, and M.~Skopina.
\newblock {\em Wavelet theory}, volume 239.
\newblock American Mathematical Soc., 2011.

\bibitem{pinkus2013smoothness}
A~Pinkus.
\newblock Smoothness and uniqueness in ridge function representation.
\newblock {\em Indagationes Mathematicae}, 24(4):725--738, 2013.

\bibitem{roy1989esprit}
R.~Roy and T.~Kailath.
\newblock Esprit-estimation of signal parameters via rotational invariance
  techniques.
\newblock {\em Acoustics, Speech and Signal Processing, IEEE Transactions on},
  37(7):984--995, 1989.

\bibitem{rudin1986real}
W.~Rudin.
\newblock {\em Real and complex analysis (3rd)}.
\newblock New York: McGraw-Hill Inc, 1986.

\bibitem{schiebinger2015superresolution}
G.~Schiebinger, E.~Robeva, and B.~Recht.
\newblock Superresolution without separation.
\newblock {\em arXiv preprint arXiv:1506.03144}, 2015.

\bibitem{schmidt1986multiple}
R.O. Schmidt.
\newblock Multiple emitter location and signal parameter estimation.
\newblock {\em Antennas and Propagation, IEEE Transactions on}, 34(3):276--280,
  1986.

\bibitem{stoica2005spectral}
P.~Stoica and R.L Moses.
\newblock {\em Spectral analysis of signals}.
\newblock Pearson/Prentice Hall Upper Saddle River, NJ, 2005.

\bibitem{just_dis}
G.~Tang, B.N. Bhaskar, and B.~Recht.
\newblock Sparse recovery over continuous dictionaries-just discretize.
\newblock In {\em Signals, Systems and Computers, 2013 Asilomar Conference on},
  pages 1043--1047, Nov 2013.

\bibitem{tang2013near}
G.~Tang, B.N. Bhaskar, and B.~Recht.
\newblock Near minimax line spectral estimation.
\newblock {\em Information Theory, IEEE Transactions on}, 61(1):499--512, 2015.

\bibitem{tang2012compressive}
G.~Tang, B.N. Bhaskar, P.~Shah, and B.~Recht.
\newblock Compressed sensing off the grid.
\newblock {\em Information Theory, IEEE Transactions on}, 59(11):7465--7490,
  2013.

\bibitem{tang2013atomic}
G.~Tang and B.~Recht.
\newblock Atomic decomposition of mixtures of translation-invariant signals.
\newblock {\em IEEE CAMSAP}, 2013.

\bibitem{terekhin2009affine}
P.A Terekhin.
\newblock Affine synthesis in the space ${L}_2({R}^d)$.
\newblock {\em Izvestiya: Mathematics}, 73(1):171, 2009.

\bibitem{tur2011innovation}
R.~Tur, Y.C. Eldar, and Z.~Friedman.
\newblock Innovation rate sampling of pulse streams with application to
  ultrasound imaging.
\newblock {\em Signal Processing, IEEE Transactions on}, 59(4):1827--1842,
  2011.

\bibitem{vetterli2002sampling}
M.~Vetterli, P.~Marziliano, and T.~Blu.
\newblock Sampling signals with finite rate of innovation.
\newblock {\em Signal Processing, IEEE Transactions on}, 50(6):1417--1428,
  2002.

\bibitem{villemoes1994wavelet}
L.F. Villemoes.
\newblock Wavelet analysis of refinement equations.
\newblock {\em SIAM Journal on Mathematical Analysis}, 25(5):1433--1460, 1994.

\bibitem{wagner2012compressed}
N.~Wagner, Y.~C Eldar, and Z.~Friedman.
\newblock Compressed beamforming in ultrasound imaging.
\newblock {\em Signal Processing, IEEE Transactions on}, 60(9):4643--4657,
  2012.

\bibitem{zhang2006schur}
F.~Zhang.
\newblock {\em The Schur complement and its applications}, volume~4.
\newblock Springer Science \& Business Media, 2006.

\end{thebibliography}
  \end{document}